\documentclass[11pt]{article}

\usepackage{amsmath,amsthm,amsfonts,amssymb,mathtools}
\usepackage{float}
\usepackage[svgnames,dvipsnames]{xcolor}
\usepackage[shortlabels]{enumitem}
\usepackage{graphicx}
\usepackage{epstopdf}
\usepackage{tikz-cd}
\usepackage{tcolorbox}
\usepackage{amsbsy,bm}
\usepackage{mathrsfs}

\usepackage[thinlines]{easytable}
\usepackage{array}

\definecolor{dullmagenta}{RGB}{102,0,102}
\def\col{dullmagenta} 

\usepackage[hidelinks,
    colorlinks=true,       
    linkcolor=\col,          
    citecolor=\col,      
    filecolor=\col,       
    urlcolor=\col,            
    ]{hyperref}

\makeatletter 
\renewcommand\paragraph{\@startsection{paragraph}{4}{\z@}%
                                      {\parskip}
                                      {-1em}%
                                      {\normalfont\normalsize\bfseries}}
\makeatother

\usepackage{endnotes}

\let\footnote=\endnote

\def\clift#1{#1^{\scriptscriptstyle\mathrm{C}}}

\parskip=\medskipamount
\newcommand{\lag}{\mathfrak{g}}

\def\lagk{\lag^k}
\def\lagdk{(\lag^*)^k}

\newcommand{\F}{\pmb{F}}

\newcommand{\R}{\mathbb{R}}
\newcommand{\Ro}{\mathcal{R}}
\newcommand{\Ac}{\mathfrak{A}}
\newcommand{\Bc}{B}

\def\fpd#1#2{{\displaystyle\frac{\partial #1}{\partial #2}}}

\newcommand{\C}{\mathbb{C}}

\newcommand{\abs}[1]{\left\lvert#1\right\rvert}

\newcommand{\st}{\;\ifnum\currentgrouptype=16 \middle\fi|\;}

\makeatletter
\def\smallunderbrace#1{\mathop{\vtop{\m@th\ialign{##\crcr
   $\hfil\displaystyle{#1}\hfil$\crcr
   \noalign{\kern3\p@\nointerlineskip}%
   \tiny\upbracefill\crcr\noalign{\kern3\p@}}}}\limits}

\makeatletter 
\def\@endtheorem{\endtrivlist}
\makeatother

\makeatletter
\def\th@plain{%
  \thm@notefont{}
  \itshape 
}
\def\th@definition{%
  \thm@notefont{}
  \normalfont 
}
\makeatother

\theoremstyle{plain}
\newtheorem{theorem}{Theorem}
\newtheorem{lemma}{Lemma}
\newtheorem{proposition}{Proposition}

\theoremstyle{definition}
\newtheorem{remark}{Remark}

\newtheorem{definition}{Definition}

\begin{document}

\date{}
\title{Cotangent bundle reduction and Routh reduction for polysymplectic manifolds}

\author{S.\ Capriotti\textsuperscript{a,b}, V. D\'iaz\textsuperscript{a},   E.\ Garc\'{\i}a-Tora\~{n}o Andr\'{e}s\textsuperscript{a,b}, T.\ Mestdag\textsuperscript{c,d}\\[2mm]
{\small \textsuperscript{a} Departamento de Matem\'atica,
Universidad Nacional del Sur (UNS),}  \\
{\small  Av.\ Alem 1253, 8000 Bah\'ia Blanca, Argentina}\\[1mm]
{\small \textsuperscript{b} Instituto de Matem\'atica de Bah\'ia Blanca UNS-CONICET, }\\
{\small  Universidad Nacional del Sur,}\\
{\small  Av.\ Alem 1253, 8000 Bah\'ia Blanca, Argentina}\\[1mm]
{\small \textsuperscript{c} Department of Mathematics,  University of Antwerp,}\\
{\small Middelheimlaan 1, 2020 Antwerpen, Belgium}\\[1mm]
 {\small \textsuperscript{d} Department of Mathematics: Analysis, Logic and Discrete Mathematics,  Ghent University,}\\
{\small Krijgslaan 281, 9000 Gent, Belgium}\\
}

\maketitle

   \begin{abstract}
We discuss Lagrangian and Hamiltonian field theories that are invariant under a symmetry group. We apply the polysymplectic reduction theorem for both types of field equations and we investigate aspects of the corresponding reconstruction process. We identify the polysymplectic structures that lie at the basis of cotangent bundle reduction and Routh reduction in this setting and we relate them by means of the Routhian function and its associated Legendre transformation.  Throughout the paper we provide   examples that illustrate various aspects of the results.

 \vspace{3mm}


 \vspace{3mm}

 \end{abstract}

\section{Introduction}

Ever since the inception of the Einstein's relativistic field equations, mathematical models for Lagrangian and Hamiltonian field theories have been at the forefront of research in mathematical physics. In this paper $k$ stands for the number of parameters on which the fields depend (e.g. $k=4$ in the case of spacetime), and we will always assume that the fields take values in a configuration manifold $Q$. In a few words, the polysymplectic model we will use in this paper, first introduced  by G\"{u}nther \cite{Gunter},  makes use of a family of $k$ closed two-forms and it characterizes the  field theory in terms of $k$ vector fields on a suitable vector bundle over $Q$. Polysymplectic geometry has proven to be appropriately applicable to many different contexts (see e.g.\  the references in the recent monograph \cite{bookpoly} on this topic). Recent contributions, both on the side of the  mathematical fundamentals \cite{McClain1, Blacker1,Regularizedpolysymplectic} and on the side of the physics applications \cite{McClain2,Blacker2,Palatinisecondorder}, testify that polysymplectic geometry is today an active field of research.

The recent literature on the geometry of field theories has seen an ever increasing interest in symmetries and  in symmetry reduction of the field equations (see e.g.\ \cite{Castrillon0,Castrillon1,Castrillon2,GbRat,LagPoincare_JGP,Ellis2,Routhfields,blacker2020reduction} for a non-exhaustive list, in different geometric frameworks, and for both Lagrangian and Hamiltonian field theories). In this context, symmetry is understood as the invariance of the Lagrangian or Hamiltonian function under the appropriate lift of the action of a Lie group $G$ on the configuration manifold $Q$. Specifically for polysymplectic structures, we will build in this paper on earlier results of e.g.\  \cite{Polyreduction,MunSal,LTM_LP,Symmetriesk}.

Since field theories are a generalization of classical mechanical systems, our methodology is mostly inspired by the many contributions on this topic in geometric mechanics. Symmetry reduction of Hamiltonian systems can best be understood in terms of either Poisson reduction  or symplectic (and cotangent bundle) reduction (see e.g. \cite{stagesbook,MarsdenRatiuIntroduction}). In the last decades it has become clear that there exist two analoguous types of symmetry reduction for Lagrangian systems: Lagrange-Poincar\'e reduction \cite{LagRedbyStag} and Routh reduction \cite{RouthMarsden,Routh_CrampinTom,quasi}. The extension of the Lagrange-Poincar\'e reduction to polysymplectic Lagrangian field theories has been investigated in \cite{LTM_LP}. In this paper we will explore both cotangent bundle reduction and Routh reduction for polysymplectic manifolds. As far as we are aware, these two specific types of reduction have not been studied elsewhere. There exist, however, results on Routh reduction for field theories in other frameworks. For example,  the paper \cite{Routhfields} investigates aspects of  the field theory version of Routh reduction in a variational framework.  
 
We start Section~\ref{sec:Preliminaries} by recalling both the standard polysymplectic structure on the cotangent bundle of $k^1$-covelocities $(T^1_k)^*Q$ for a Hamiltonian field theory, and the Poincar\'e-Cartan polysymplectic structure that leads to the $k$-symplectic Lagrangian field equations, and we provide some physical and geometric examples in Section~\ref{subsec:someexamples}.  We then state, in Section~\ref{polyredsection}, for a general polysymplectic structure with momentum map $J$, a version of the so-called polysymplectic reduction theorem of \cite{Polyreduction}: we mention both the aspects concerning the reduced polysymplectic form (Theorem~\ref{thm:polyred}) and those concerning the reduced dynamics (Theorem~\ref{thm:polyred2}). We would like to point out that this type of reduction only applies under rather restrictive conditions, which are very specific to the field theory case and which do not show up in the case of a classical mechanical system. In classical mechanics, we know via Noether's theorem that there is (roughly speaking) a correspondence between the invariance of the system under symmetry groups and conservation laws, and that the symplectic reduction theorem heavily relies on this property. This correspondence is no longer valid in the case of a Lagrangian field theory: there the invariance under a symmetry group only results in the vanishing of a divergence, and in general no first integrals can be derived from this property. The  polysymplectic reduction theorem of \cite{Polyreduction} circumvents this problem: one needs to ask for the dynamics to be tangent to a level set of momentum, as in Theorem~\ref{thm:polyred2}. This is the reason why this property also appears in our Theorem~\ref{thm:red1}. We will see in the final section, however, that this does not prevent us from discussing some explicit solutions of interesting examples and applications.

In \cite{Polyreduction} there is no description of the reconstruction procedure after reduction. We discuss in Section~\ref{recsec} first the notion of a $k$-vector field and its integrability and we give a new perspective on how and when the integrability of an invariant $k$-vector field can be inherited by its reduced $k$-vector field (Proposition~\ref{pro:reconstruction}). Based on results of \cite{LTM_LP}, we state in Theorem~\ref{thm:reconstruction} how the integral sections of the original field theory can be reconstructed from those of the reduced field theory.

As stated before, one of the main goals of this paper is to extend Routh reduction to polysymplectic field theories. We will base our methodology on geometric constructions that are similar to those in \cite{quasi} for the mechanical case. That is: we will explain Routh reduction via two distinct steps. First, in Section~\ref{sec:cotangentreduction} we refine the polysymplectic reduction theorem to the specific case of the $k$-cotangent bundle $(T^1_k)^*Q$ and its canonical momentum map $J$. This results in Theorem~\ref{thm:identification}, which is our analogue of the cotangent bundle reduction theorem (see e.g.\ \cite{stagesbook}). This Theorem has an interest on its own because, together with Theorem~\ref{thm:polyred2}, we obtain in this way the complete picture of the polysymplectic reduction of a Hamiltonian system on the standard polysymplectic manifold $(T^1_k)^*Q$.

In the second step we introduce the momentum map $J_L$ of a Lagrangian field theory  and we use Lemma~\ref{lem:lema1} to identify the reduced manifold $J_L^{-1}(\mu)/G_\mu$ of the Lagrangian polysymplectic structure in Proposition~\ref{prop:identificationJL}. With the help of the so-called Routhian function, we show how its Legendre-type transformation relates the two reduced polysymplectic structures of $J_L^{-1}(\mu)/G_\mu$ and $J^{-1}(\mu)/G_\mu$. By pulling back the geometric structures from the Hamiltonian side, we can summarize in Theorem~\ref{thm:reducedform} the discussion with an explicit expression of the reduced polysymplectic structure on $J_L^{-1}(\mu)/G_\mu$ for the Lagrangian side. Finally, in Theorem~\ref{thm:red1} we relate the reduced Lagrangian field equations to the reduced polysymplectic form on $J_L^{-1}(\mu)/G_\mu$.  The picture is then completed with a version of the reconstruction problem in Theorem~\ref{thm:reconstruction2}. 

The paper ends with some examples and applications in Section~\ref{ex2sec}: we revisit the example of Navier's equations from Section~\ref{subsec:someexamples}, we discuss the case where the configuration manifold is a Lie group, and we illustrate the application to harmonic maps by means of a concrete example.

\section{Preliminaries}\label{sec:Preliminaries}

\subsection{Polysymplectic manifolds}

We will use Einstein's conventions for sums over $i=1,\dots,{\rm dim(Q)}$, but no sum signs for sums over $a=1,\dots,k$.

Let $N$ be a manifold of dimension $n(k+1)$. A \emph{$k$-polysymplectic structure} in $N$ is a closed non-degenerate $\R^k$-valued two-form $\Omega$, say
\begin{equation*}
\Omega= \sum_a \omega^a \otimes e_a, 
\end{equation*}
with $1\leq a\leq k$, where $\omega^a$ are two-forms on $N$ and $\{e_1,\dots,e_k\}$ is a basis of $\R^k$. Clearly, having a $k$-polysymplectic structure is equivalent to the existence of a family of $k$ closed two-forms $(\omega^1,\dots,\omega^k)$ such that $\ker\omega^1\cap\dots\cap\ker \omega^k=\{0\}$. The case $k=1$ corresponds to a symplectic manifold; however, for $k>1$, we do not have in general a family of $k$ symplectic forms, since each of the two-forms $\omega^1,\dots,\omega^k$ might be degenerate. Introduced by G\"{u}nther in~\cite{Gunter}, polysymplectic structures provide a simple formalism to study a broad class of field theories, specifically those for which the Lagrangian or Hamiltonian is explicitly independent of the space-time coordinates.

A first and important example of polysymplectic manifold is the cotangent bundle of $k^1$-covelocities
\[
(T^1_k)^*Q =T^*Q\oplus\stackrel{k}{\dots}\oplus T^*Q
\]
equipped with the family of two-forms $(\omega_Q^1,\dots,\omega_Q^k)$ defined by $\omega_Q^a=(\pi_Q^a)^*\omega_Q$, where $\omega_Q$ the canonical symplectic form on $T^*Q$ and $\pi_Q^a\colon (T^1_k)^*Q\to T^*Q$ is the projection defined by
\[
\pi_Q^a(\alpha_q^1,\dots, \alpha_q^k)= \alpha_q^a.
\]
Elements in $(T^1_k)^*Q$ will be denoted using bold face characters, such as $\bm{\alpha}_q=(\alpha_q^1,\dots, \alpha_q^k)$ and, when the basepoint $q$ is understood, we will drop it and write $\bm{\alpha}=(\alpha^1,\dots,\alpha^k)$. If $q^i$ are coordinates on an open set $U\subset Q$, with $i=1,\dots,\dim Q$, then there are induced coordinates $(q^i,p_i^a)$ on $(\pi_Q^a)^{-1}(U)$ with $p_i^a$ the components of $\alpha_q^a$ in the usual basis of $T_q^*Q$. Our convention is to choose $\omega_Q=-d\theta_Q$ with $\theta_Q= p_i dq^i$ the tautological one-form on $T^*Q$, and therefore in the coordinates $(q^i,p_i^a)$ we have
\begin{equation}\label{eq:omega-a}
\omega_Q^a= dq^i\wedge dp_i^a=-d\theta^a_Q,  
\end{equation}
where $\theta^a_Q=(\pi_Q^a)^*\theta_Q= p_i^a dq^i$.

The cotangent bundle of $k^1$-covelocities plays the role of the phase space for the Hamiltonian field theory in the $k$-symplectic framework. The dual of this space, which plays a similar role in the Lagrangian picture, is the  tangent bundle of $k^1$-covelocities. It is denoted 
\[
T^1_kQ=TQ\oplus\stackrel{k}{\dots} \oplus TQ,
\]
and can be identified with $J^1_0(\R^k,Q)$, the manifold of 1-jets of maps from $\R^k$ to $Q$ with source at $0\in\R^k$. An element in $T^1_kQ$ is denoted $\pmb{v}_q=(v_{1q},\dots,v_{kq})$ or, more often, simply $\pmb{v}=(v_{1},\dots,v_{k})$. Given coordinates $q^i$ on $Q$, the induced coordinates on $T^1_kQ$ are $(q^i,v^i_a)$, with $v^i_a$ the components of $v_{aq}$ in the natural basis of $T_qQ$. The projection of $T^1_kQ$ onto $Q$ will be denoted by $\tau_Q^k\colon T^1_kQ\to Q$. The natural pairing between elements of $(T^1_k)^*Q$ and $T^1_kQ$ is written as
\[
\langle\bm{\alpha}_q,\pmb{v}_q\rangle=\langle\alpha_q^1,v_{1q}\rangle +\dots+\langle \alpha_q^k,v_{kq}\rangle. 
\]

A section $\pmb{X}$ of $\tau_Q^k\colon T^1_kQ\to Q$ is called a $k$-vector field on $Q$. If $\tau_Q^{k,a}\colon T^1_kQ\to TQ$ denotes the projection on the $a$-th component, we will denote
\[
X_a= \tau_Q^{k,a}\circ \pmb{X}\colon Q\to TQ,
\]
and write
\[
\pmb{X}=(X_1,\dots,X_k). 
\]
Therefore, we might think of a $k$-vector field $\pmb{X}$ on $Q$ as a family of $k$ (standard) vector fields $X_a$ on $Q$.

We will use $t=(t^1,\dots,t^k)$ to denote points in $\R^k$. For any map $f\colon N\to M$ between manifolds, we write $T_nf$ or simply $Tf$  for the tangent map of $f$ at $n\in N$.

\begin{definition}\label{def:integrable} An \emph{integral section} of $\pmb{X}$ through $p\in Q$ is a map $\varphi\colon U\subset\R^k\to Q$, defined on some neighbourhood $U\ni 0$, such that
\[
\varphi(0)=p,\qquad T_t\varphi\circ \left(\left.\fpd{}{t^a}\right|_t\right)=X_a(\varphi(t)),\quad \text{for every $t\in U$}. 
\]
We say that $\pmb{X}$ is \emph{integrable} if there is an integrable section through every point of $Q$. 
\end{definition}

We will usually take $U=\R^k$ in the definition of integral section above. For any map 
\[
\phi=(\phi^1,\dots,\phi^{\dim(Q)})\colon U\subset \R^k\to Q, 
\]
the first prolongation is the map $\phi^{(1)}\colon U\to T^1_kQ$ given by
\[
\phi^{(1)}(t)=\left(\phi(t),T_t\phi\,(\partial/\partial t^1),\dots,T_t\phi\,(\partial/\partial t^k)\right).
\]
An integral section of $\pmb{X}$ is then a map which satisfies  $\varphi^{(1)}(t)=\pmb{X}\circ \varphi(t)$, or
\begin{equation}\label{eq:integralsection}
\fpd{\varphi^i}{t^a}=X^i_a\circ\varphi, 
\end{equation}
where $X_a= X_a^i\,\partial/\partial q^i$. From Definition~\ref{def:integrable}, one can show that in order for $\pmb{X}$ be integrable one needs to rectify the family of vector fields $X_1,\dots,X_k$. It is well-known that this is possible iff the condition $[X_a,X_b]=0$ holds for each $a,b$. This condition can be expressed as
\[
X_a^i\fpd{X_b^j}{q^i}- X_b^i\fpd{X_a^j}{q^i}=0,
\]
which is precisely the integrability condition of the PDE~\eqref{eq:integralsection} when one applies the chain rule. 

\begin{remark}Even though the terminology ``polysymplectic structure'' and ``$k$-symplectic structure'' is often used interchangeably in the literature, these two terms actually denote different geometric structures. The definition of a polysymplectic manifold we use here is due to G\"{u}nther \cite{Gunter} and it agrees, among others, with \cite{PolyPoisson} and \cite{Polyreduction}. A polysymplectic manifold consists of a manifold $M$ equipped with a family of $k$ closed two-forms $\omega^a$ satisfying a nondegeneracy condition. An important class of polysymplectic manifolds are the so-called \emph{standard polysymplectic structures}~\cite{Gunter}: a polysymplectic manifold is standard if it is  locally polysymplectomorphic to the cotangent bundle of $k^1$-velocities.  The term $k$-symplectic manifold was coined by Awane~\cite{Awane} and, at least originally, it refers to a polysymplectic manifold (in the sense above) of dimension $n(k+1)$ for some $n$ (here $k$ has the same meaning as above), which is further equipped with an integrable distribution $V$ of dimension $nk$ that satisfies $\omega^a(V,V)=0$. Clearly, a $k$-symplectic structure is a polysymplectic structure (actually, a standard one), but the converse is not true in general. We also recall that $k$-symplectic structures appeared independently in \cite{palmoscot} under the name of \emph{k-almost cotangent structures}.
 In the context of symmetry reduction~\cite{Polyreduction}, it is most convenient to work with polysymplectic manifolds, rather than $k$-symplectic,  because the reduction of a $k$-symplectic manifold need not be a $k$-symplectic manifold.

Unfortunately, there is even more possible confusion in terminology and we need a further disclaimer. The term ``polysymplectic structure'' is also used in a series of  papers, starting with e.g.\ \cite{Gia,Kana}. There it stands for a certain vector-valued form defined on an associated  bundle of a given fibre bundle. In this paper we do {\em not} make use of this different description of classical field theories.   \end{remark}

\subsection{Hamiltonian and Lagrangian \texorpdfstring{$k$}{k}-symplectic field theory}\label{sec:HLfieldtheory}

Given a polysymplectic manifold $(N,\omega^a)$, the vector bundle morphism $\flat_\omega\colon T^1_k N\to T^*N$ is defined as
\[
\flat_\omega(v_1,\dots,v_k)= \sum_a v_a\lrcorner \omega^a.
\]
It can be easily shown that $\flat_\omega$ is surjective: indeed, it is the dual -up to a sign- of the injective map over the identity $TN\to (T^1_k)^*N$ given by $v\mapsto (v\lrcorner\omega^1,\dots,v\lrcorner \omega^k)$.
 
Consider a Hamiltonian $H\colon (T_k^1)^*Q\to \R$. The $k$-symplectic Hamiltonian eqs. for a $k$-vector field $\pmb{X}$ on $(T_k^1)^*Q$ are
\begin{equation}\label{eq:k-HAM}
 \flat_{\omega_Q} (\pmb{X})=dH, \tag{k-HAM}
\end{equation}
i.e.\ $\sum_a X_a\lrcorner \omega_Q^a=dH$ ($\omega_Q^a$ have been defined in~\eqref{eq:omega-a}). Since $ \flat_{\omega_Q}$ is surjective, the equations~\eqref{eq:k-HAM} always admit a solution which, in general, will not be unique. The main interest of these equations is that the integrable sections of $\pmb{X}$ give solutions of the familiar Hamilton-De\,~Donder-Weyl equations in Hamiltonian field theory (cf. the book \cite{bookpoly} for a full discussion). One finds in this way the so-called admissible solutions~\cite{RomanRey}. 

The equations~\eqref{eq:k-HAM} are a particular case of a more general class of $k$-symplectic Hamiltonian equations: if $(N,\omega^a)$ is a polysymplectic manifold and $H\colon N\to\R$ is the Hamiltonian, the $k$-symplectic Hamiltonian eqs. for a $k$-vector field $\pmb{X}$ on $N$ are
\begin{equation}\label{eq:k-Sym}
\flat_{\omega} (\pmb{X})=dH. \tag{k-Sym}
\end{equation}
Of course, one recovers~\eqref{eq:k-HAM} in the case $N=(T_k^1)^*Q$, and we speak of $k$-symplectic Hamiltonian equations for both~\eqref{eq:k-HAM} and ~\eqref{eq:k-Sym}.

We now recall how the Euler-Lagrange equations in Lagrangian field theory can be obtained in the $k$-symplectic setting; a detailed discussion, including a comprehensive list of references, can be found in~\cite{bookpoly}. The Euler-Lagrange field equations for a Lagrangian $L\colon T^1_k Q\to\R$ of the form $L(q^i,v^i_a)$ are the following set of second-order PDEs:
\begin{equation}\label{eq:EL}
\sum_a \left.\fpd{}{t^a}\right|_t \left(\left.\fpd{L}{v^i_a}\right|_{\varphi(t)}\right)=\left.\fpd{L}{q^i}\right|_{\varphi(t)},\quad \left.v^i_a(\varphi(t))\right|_{\varphi(t)} =\left.\fpd{\varphi^i}{t^a}\right|_t, \tag{EL}
\end{equation}
where $\varphi=(\varphi^1,\dots,\varphi^k)\colon \R^k\to T^1_kQ$. We are interested in  describing solutions of~\eqref{eq:EL} for a regular Lagrangian using a $k$-symplectic approach. This can be done in a way that closely resembles the case of a regular Lagrangian in classical mechanics, where the Legendre transformation is used to construct a symplectic form and an energy function that govern the dynamics.

Given a Lagrangian $L\colon T^1_k Q\to\R$ , the Legendre transformation is the map 
\[
\pmb{F} L\colon T^1_kQ\to (T^1_k)^*Q 
\]
with components $(FL)^a\colon T^1_kQ\to T^*Q $ given by
\begin{equation}\label{eq:Legendre}
\langle(FL)^a(\pmb{v}),w_q\rangle = \left.\frac{d}{ds}\right|_{s=0}L(v_{1q},\dots,v_{aq}+sw_q,\dots,v_{kq}),\qquad w_q\in TQ.
\end{equation}
In coordinates, it reads
\[
\pmb{F} L(q^i,v^i_a)=\left(q^i,\fpd{L}{v^i_a}\right). 
\]
The energy function associated to the Lagrangian $L$, $E_L\colon T^1_kQ\to\R$, is defined as
\begin{equation*}
E_L(\pmb{v})=\langle \F L (\pmb{v}),\pmb{v}\rangle -L(\pmb{v}),
\end{equation*}
or in coordinates:
\[
E_L(q^i,v^i_a)=\sum_{a} v^i_a\fpd{L}{v^i_a}-L.
\]

\begin{definition} A Lagrangian $L\colon T^1_k Q\to\R$ is \emph{regular} if $\pmb{F} L$ is a local diffeomorphism. If, additionally, $\pmb{F} L$ is a global diffeomorphism then $L$ is \emph{hyperregular}.
\end{definition}

Locally, the Lagrangian $L$ is regular if the matrix
\begin{equation*}
\left(\fpd{^2 L}{ v^i_a\partial v^j_b}\right)\qquad 1\leq i,j\leq n,\;\;  1\leq a,b\leq k,
\end{equation*}
has maximal rank everywhere on $T^1_kQ$ (this maximal rank equals $kn$). When the Lagrangian is regular, the family of 2-forms 
\[
\omega^a_{Q,L}=(\F L)^*\omega^a_Q= dq^i\wedge d\left(\fpd{L}{v^i_a}\right),\qquad a=1,\dots,k, 
\]
are a polysymplectic structure on $T^1_kQ$. In this case, we will consider the following set of equations for a $k$-vector field $\pmb{\Gamma}$ on $T_k^1Q$:
\begin{equation}\label{eq:k-EL}
\sum_a \Gamma_a\lrcorner \omega_{Q,L}^a=dE_L, \tag{k-EL}
\end{equation}
which will be referred to as \emph{(k-symplectic) Euler-Lagrange equations} (or k-EL for short). They are of course a particular case of~\eqref{eq:k-Sym}. It can be shown that, defining the Hamiltonian as $H=E_L\circ (\F L)^{-1}$, the Legendre transformation gives a bijection between the solutions of \eqref{eq:k-HAM} and \eqref{eq:k-EL} and their integral sections; for a proof, see \S 6.4 in~\cite{bookpoly}.  

The importance of~\eqref{eq:k-EL} lies in the following result:

\begin{theorem}\label{thm:regularsolutions} Let $L$ be a regular Lagrangian and $\pmb{\Gamma}$ be an integrable solution of \eqref{eq:k-EL}. Then the integral sections $\varphi\colon Q\to T^1_kQ$ of $\pmb{\Gamma}$ are prolongations of maps $\phi\colon \R^k\to Q$ and are solutions of \eqref{eq:EL}.   
\end{theorem}

An integrable solution $\pmb{\Gamma}$ of \eqref{eq:k-EL} will be called an integrable Lagrangian SOPDE. In other words, Theorem~\ref{thm:regularsolutions} states that, if an integrable Lagrangian SOPDE for a given regular Lagrangian $L$ is known, we can find solutions of the Euler-Lagrange equations of $L$ by finding integral sections $\varphi=\phi^{(1)}\colon Q\to T^1_kQ$ of $\Gamma$. By definition, the projection $\phi=\tau_Q^k\circ\phi^{(1)}\colon \R^k\to Q$ satisfies the Euler-Lagrange equations
\begin{equation}\label{eq:ELphi}
\sum_a \left.\fpd{}{t^a}\right|_t \left(\left.\fpd{L}{v^i_a}\right|_{\phi^{(1)}(t)}\right)=\left.\fpd{L}{q^i}\right|_{\phi^{(1)}(t)}.
\end{equation}
We will also say that a map $\phi\colon \R^k\to Q$ is a solution of~\eqref{eq:EL} if it satisfies the equations~\eqref{eq:ELphi} above. The situation is similar to the Hamiltonian case: here one finds the admissible solutions of the Euler-Lagrange equations, namely those solutions that can be retrieved as an integral section of a some integral solution $\pmb{\Gamma}$, see~\cite{RomanRey}. 

The term SOPDE stands for ``second order partial-differential equation'', which means that $\pmb{\Gamma}=(\Gamma_1,\dots, \Gamma_k)$ is of the form
\begin{equation}\label{eq:SOPDE}
\Gamma_a=v^i_a \fpd{}{q^i}+ \sum_b (\Gamma_a)^i_b \fpd{}{v^i_b},
\end{equation}
for some functions on $(\Gamma_a)^i_b$ on $T^1_kQ$. This is equivalent with the fact that all the integral curves of $\Gamma$ are prolongations of maps $\phi\colon \R^k\to Q$ (see e.g.~\cite{MunSal} for a proof). The solutions of~\eqref{eq:k-EL} for regular Lagrangians are always of this type.  By definition, an integral section $\varphi=(\varphi^i,\varphi^i_a)\colon \R^k\to T^1_kQ$ of the SOPDE~\eqref{eq:SOPDE} satisfies 
\[
\fpd{\varphi^i}{t^a}= \varphi^i_a,\qquad \fpd{\varphi^i_a}{t^b}= (\Gamma_a)^i_b\circ \varphi,
\]
which may as well be written as a system of second order partial differential equations:
\[
\fpd{^2\varphi^i}{t^a\partial t^b}(t)= (\Gamma_a)^i_b\left(\varphi^i(t),\fpd{\varphi^i(t)}{t^c}\right).
\]
In particular, the equality of the mixed partial derivatives gives the integrability condition $(\Gamma_a)^i_b=(\Gamma_b)^i_a$, valid for any integrable SOPDE $\pmb{\Gamma}$. 

Later, when we study the reduced Lagrangian field theories we will need extensions of some of the notions above. In particular the reduced space will be a pullback bundle $\Pi^* T^1_kQ$ for some bundle $\Pi\colon P\to Q$. 
In this case:
\begin{enumerate}[label={(\roman*})]
\item An element of $\Pi^* T^1_kQ$ will be denoted by $\pmb{v}=(p,v_{1q},\dots,v_{kq})$, with $\Pi(p)=q$.
\item A Lagrangian on the reduced space is a function $L\colon \Pi^* T^1_kQ\to \R$.
\item The Legendre transformation is a map $\pmb{F} L\colon \Pi^* T^1_kQ\to \Pi^*(T^1_k)^*Q$. If $q^i$ are coordinates on $Q$ and $(q^i,y^\alpha)$ are bundle coordinates on $P$, the Legendre transformation is simply  
\[
\pmb{F} L(q^i,y^\alpha,v^i_a)=\left(q^i,y^\alpha,\fpd{L}{v^i_a}\right). 
\]
\item The energy $E_L\colon \Pi^* T^1_kQ\to\R$ is the function
\[
E_L(q^i,y^\alpha, v^i_a)=\sum_{a} v^i_a\fpd{L}{v^i_a}-L. 
\]
\end{enumerate}
The intrinsic definition of $\pmb{F} L$ and $E_L$ is obtained from the natural pairing between $\Pi^*(T^1_k)^*Q$ and $\Pi^* T^1_kQ$ as in the standard case.

\subsection{Some examples}\label{subsec:someexamples}

First-order field theories whose Hamiltonian or  Lagrangian does not depend explicitly on the space of parameters can be studied within the polysymplectic formalism. We collect here a few examples that we will use to illustrate some of our results in Section~\ref{ex2sec}.

\paragraph{Laplace equation.} Consider the Lagrangian $L\colon T^1_k\R \to\R$ given by
\[
L(q,v_1,\dots,v_k)=\frac{1}{2}\left((v_1)^2+\dots+(v_k^2)\right). 
\]
An integral section $\phi^{(1)}\colon \R^k\to T^1_k\R$ of a $k$-vector field for this Lagrangian satisfies the Laplace equation
\begin{equation} \label{Lapeq}
\fpd{^2  \phi}{(t^1)^2}+\dots+ \fpd{^2 \phi}{(t^k)^2}=0.
\end{equation}

\paragraph{Navier's equations.} The Euler-Lagrange equations for the Lagrangian $L\colon T^1_2\R^2\to \R$
\begin{equation}\label{eq:NavierLagrangian}
L(q^i,v^i_a)= \left(\frac{\lambda}{2}+\nu \right)\left[(v^1_1)^2+(v^2_2)^2\right]+\frac{\nu}{2}\left[(v^1_2)^2+(v^2_1)^2\right]+(\lambda+\nu)v^1_1v^2_2
\end{equation}
lead to the following system  of PDEs:
\begin{align}\label{eq:Navier}
\begin{split}
(\lambda+2\nu)\partial_{11}\phi^1+(\lambda+\nu)\partial_{12}\phi^2+\nu\partial_{22}\phi^1&=0,\\
\nu\partial_{11}\phi^2+(\lambda+\nu)\partial_{12}\phi^1+(\lambda+2\nu)\partial_{22}\phi^2&=0.
\end{split}
\end{align}
These are the so-called Navier's equations in two dimensions, which are the equations of linear isotropic elasticity~\cite{OlverBook}. Here $\lambda$ and $\nu$ are constants, known as \emph{Lam\'e moduli} (the parameter $\nu$ is usually denoted $\mu$, but we will use the symbol $\mu$, later, for the momentum). The $4\times 4$ matrix
\[
\begin{pmatrix}
\fpd{^2 L}{ v^i_1\partial v^j_1} & \fpd{^2 L}{ v^i_2\partial v^j_1} \\[1.5em]
\fpd{^2 L}{ v^i_1\partial v^j_2}  & \fpd{^2 L}{ v^i_2\partial v^j_2} 
\end{pmatrix} =
\begin{pmatrix}
\lambda + 2\nu & 0 & 0 & \lambda + \nu \\
0 & \nu & 0 & 0 \\
0 & 0& \nu & 0 \\
\lambda + \nu & 0 & 0 & \lambda + 2\nu \\
\end{pmatrix} 
\]
has determinant $\nu^3(2\lambda+3\nu)$. We will assume from now that both $\nu\neq0$ and $2\lambda+3\nu\neq 0$, so that the Lagrangian is regular. The symmetries of this Lagrangian (in the more general $3$-dimensional case) have been studied in detail in~\cite{OlverElasticityII}. We will come back to this example in Section~\ref{exNE}.

\paragraph{Harmonic maps from $\R^n$.} Consider a harmonic map $\varphi\colon (M,h_{ab})\to (Q,g_{ij})$ between two Riemannian manifolds (see \cite{JOSTRiemannianBook} for a concrete exposition). If we set $(M,h_{ab})=(\R^k,\delta_{ab})$ with its Euclidean metric, then it is well-known that $\varphi$ appears as a solution of the following Lagrangian $L\colon T^1_kQ\to \R$:
\[
L(q^i,v^i_a)=\frac{1}{2} \sum_{a,b=1}^k g_{ij}(q^i) v^i_a v^i_b.
\]
The Lagrangian $L$ can also be written as the sum of the $k$ Lagrangians, obtained from the metric Lagrangian on $TQ$ pulled back to each of the components of $T_1^kQ$, i.e.
\[
L(q,\pmb{v}_q)=\frac{1}{2}\mathcal{G}(\pmb{v}_{1q},\pmb{v}_{1q})+\dots+ \frac{1}{2}\mathcal{G}(\pmb{v}_{kq},\pmb{v}_{kq}),
\]
with $\mathcal{G}=(g_{ij})$ the metric on $Q$. In the terminology of~\cite{bookpoly}, this is a particular case of a \emph{$k$-symplectic quadratic system}. We will discuss a concrete case of the example  in Section~\ref{exHM}.

%

\paragraph{Einstein-Palatini Gravity.} Let $Q=\Sigma\times_M \mathcal{C}(LM)$, where  we take $M=\R^4$ as a model for space-time, $\Sigma\to M$ is the bundle of Lorentzian metrics on $M$ and $\mathcal{C}(LM)\to M$ the bundle of linear connections on $M$. The Lagrangian $L\colon T^1_kQ\to\R$ is
\[
L(g_{ab},\Gamma^d_{ef},\partial_cg_{ab},\partial_c\Gamma^d_{ef})= \sqrt{\abs{{\rm det}\, g_{ab}}}\cdot R,
\]
with $R=g^{ab} R_{ab}$ the scalar of curvature ($R_{ab}$ denotes the Ricci curvature). Roughly speaking, besides the metric tensor $g_{ab}$, the connection coefficients $\Gamma^c_{ab}$ are also part of the unknowns in this model. For this reason the above Lagrangian, being affine in the derivatives of $\Gamma^c_{ab}$, is not regular. This example is studied in detail in the framework of polysymplectic geometry in~\cite{Palatinisecondorder}.

\paragraph{Complex Scalar field.} The Lagrangian $L\colon T^1_2\R^2\to \R$ of a complex scalar field $\phi\in \C\simeq \R^2$ with a quartic interaction can be written, in terms of its real and imaginary parts $\phi_1$ and $\phi_2$ as follows:
\begin{equation}\label{eq:complexsf}
L(\phi_i,v^i_a)=\frac{1}{2}\left(\sum_{a,b=1}^4 \eta_{ab}v^1_av^1_b-m^2\phi_1^2\right)+\frac{1}{2}\left(\sum_{a,b=1}^4 \eta_{ab}v^2_av^2_b-m^2\phi_2^2\right)-\frac{1}{4}g(\phi_1^2+\phi_2^2)^2.
\end{equation}
Here $g$ is a coupling constant and $\eta={\rm diag}(-1,+1,+1,+1)$ is the Minkowski metric. Complex scalar fields are often studied from the point of view of quantization. See e.g.~\cite{LecturesQFT} for more details.

\section{Polysymplectic reduction and reconstruction}\label{sec:polyred}


The aim of this section is to recall the polysymplectic reduction theorem in~\cite{Polyreduction} and to discuss, when possible, the reconstruction of solutions.


\subsection{Actions and connections}\label{subsec:connections}

If $G$ is a Lie group, we will denote by $\lag$ its Lie algebra and by ${\rm Ad}\colon G\times \lag\to\lag$ the adjoint action. The coadjoint action will be denoted ${\rm Coad}\colon G\times \lag\to\lag$, and is defined as usual by:
\[
{\rm Coad}_g(\mu) = {\rm Ad}_{g^{-1}}^*\mu. 
\]
We consider the spaces $\lagk=\lag\times \stackrel{k}{\dots} \times \lag$ and  $\lagdk = \lag^*\times \stackrel{k}{\dots} \times \lag^*$. They are equipped with the $k$-Adjoint and $k$-Coadjoint actions, defined as: 
\begin{align*}
 {\rm Ad}^k\colon G\times  \lagk &\to\lagk,\\
 (g,\xi_1,\dots,\xi_k)&\mapsto ({\rm Ad}_g \xi_1,\dots,{\rm Ad}_g \xi_k),
 \end{align*}
and
\begin{align*}
 {\rm Coad}^k\colon G\times  \lagdk &\to\lagdk,\\
 (g,\mu_1,\dots,\mu_k)&\mapsto ({\rm Ad}_{g^{-1}}^*\mu_1,\dots,{\rm Ad}_{g^{-1}}^*\mu_k),
 \end{align*}
respectively. 

In this paper we will only consider free and proper actions of a Lie group $G$ on a manifold $Q$. This guarantees that $\pi \colon Q\to Q/G$ is a principal bundle. If $\Phi\colon G\times Q\to Q$ is an action, we will denote by $\xi_Q$ the infinitesimal generator of the action corresponding to $\xi\in\lag$. For concreteness, we will work with left actions, and given $q\in Q$ we will write $\Phi_g(q)=g\cdot q$ when there is no risk of confussion. 

Given an action $\Phi$ on $Q$, there are induced $G$-actions $\Phi^{TQ}$ on $TQ$ and $\Phi^{T^1_kQ}$ on  $T^1_kQ$. They are defined as follows:
\begin{align*}
\Phi^{TQ}_g(v_q)&=T_q\Phi_g(v_q), \\
\Phi^{T^1_kQ}_g(v_{1q},\dots,v_{kq})&=(T_q\Phi_g(v_{1q}),\dots,T_q\Phi_g(v_{kq})).
\end{align*}
We will write the actions above as $g\cdot v_q$ and $g\cdot \pmb{v}_q$ respectively. There are also induced actions on $T^*Q$ and $(T^1_k)^*Q$ defined similarly as follows:
\begin{align*}
\Phi^{T^*Q}_g(\alpha_q)&=T^*_{g\cdot q}\Phi_{g^{-1}}(\alpha_q), \\
\Phi^{(T^1_k)^*Q}_g(\alpha^1_{q},\dots,\alpha^k_{q})&=(T^*_{g\cdot q}\Phi_{g^{-1}}(\alpha^1_q),\dots,T^*_{g\cdot q}\Phi_{g^{-1}}(\alpha^k_q)).
\end{align*}
We will also write $g\cdot \alpha_q$ and $g\cdot \bm{\alpha}_q$ for the actions above. Note, in particular, that the pairings satisfy $\langle \alpha_q,v_q\rangle=\langle g\cdot \alpha_q, g\cdot v_q\rangle$ and $\langle \pmb{\alpha}_q,\pmb{v}_q\rangle=\langle g\cdot \pmb{\alpha}_q, g\cdot \pmb{v}_q\rangle$.

Recall that a principal connection on the principal bundle $Q\to Q/G$ can be given in terms of a connection 1-form $\Ac\colon TQ\to \lag$ which satisfies: 
\begin{enumerate}[label={(\roman*})]
 \item $\Ac(\xi_Q)=\xi$, for each $\xi\in\lag$,
 \item $\Ac(g\cdot v)={\rm Ad}_{g} (\Ac(v))$, for each $g\in G$ and $v\in TQ$.
\end{enumerate}
Alternatively, a principal connection defines an (invariant) horizontal subbundle $\mathcal{H}\subset TQ$:
\[
\mathcal{H}=\ker(\Ac). 
\]
We will use the following definition for the curvature $K$ of $\Ac$: for two vector fields $X,Y$ on $Q$
\[
K(X,Y)=-{\rm Ver}\big([{\rm Hor}(X),{\rm Hor}(Y)]\big),
\]
where ${\rm Hor}(\cdot)$ and ${\rm Ver}(\cdot)$ denote the horizontal and vertical parts w.r.t. $\Ac$. By definition, the curvature of two vector fields is again a vector field. 

Given a principal connection $\Ac$, we will sometimes denote by $\pmb{\Ac}\colon T^1_kQ\to \lagk$ the map: 
\[
\pmb{\Ac}(v_1,\dots,v_k)=\left(\Ac(v_1),\dots,\Ac(v_k)\right).
\]
The map $\pmb{\Ac}$ is, in the terminology of \cite{LTM_LP}, a \emph{simple principal $k$-connection}.

\subsection{Polysymplectic reduction} \label{polyredsection}

Assume that $(N,\omega^a)$ is a polysymplectic manifold, and that $\Phi_g\colon N\to N$ is a polysymplectic action of a Lie group $G$ with momentum map 
\[
J=(J^1,\dots,J^k)\colon N\to \lag^*\times \stackrel{k}{\dots} \times \lag^*=\lagdk, 
\]
where $\lag$ is the Lie algebra of $G$. This means that $\Phi_g^*\omega^a=\omega^a$ (for each $a=1,\dots,k$) and that ${\xi_N}\lrcorner  \omega^a=d J^a_\xi$, with $J^a_\xi\colon N\to\R$ defined as $J^a_\xi(x)=\langle J^a(x),\xi\rangle$ for each $x\in N$. We also require equivariance of the momentum map w.r.t.\ the $k$-Coadjoint action (see Section~\ref{subsec:connections}), namely
\[
J\circ\Phi_g={\rm Coad}^k_g\circ J. 
\]
Let us denote by $G_\mu\subset G$ the isotropy group of $\mu=(\mu_1,\dots,\mu_k)$ under the ${\rm Coad}^k$ action and by $\lag_\mu$ its Lie algebra. One checks easily that the following holds:
\[
 G_\mu= G_{\mu_1}\cap\dots\cap G_{\mu_k},\qquad \lag_\mu=\lag_{\mu_1}\cap\dots\cap \lag_{\mu_k},
\] 
where $G_{\mu_a}$ is the isotropy group of $\mu_a$ under the usual coadjoint action ${\rm Coad}$ of $G$, and $\lag_{\mu_a}$ is the Lie algebra of $G_{\mu_a}$. In particular, $G_\mu$ is a subgroup of $G_{\mu_a}$ for each $a$.
 
For future reference, we need to specify two technical conditions for the momentum map:
\begin{enumerate}[label={(\roman*})]
 \item[(C1)] $\ker (T_xJ^a)=T_x(J^{-1}(\mu))+\ker\omega^a\mid_x+T_x(G_{\mu_a}\cdot x)$,\; for all $x\in J^{-1}(\mu)$, and for each $a$.
 \item[(C2)] $T_x(G_{\mu}\cdot x)=\cap_a \big[T_x(G_{\mu_a}\cdot x)+\ker\omega^a\mid_x\big]\cap T_x(J^{-1}(\mu))$,\; for all $x\in J^{-1}(\mu)$.
\end{enumerate}
\vspace{-.7pc}The notation $G_{\mu_a}\cdot x=\{g\cdot x\st g\in G_{\mu_a}\}$ stands for the orbit of $x$ under the $G_{\mu_a}$ action.
 
The polysymplectic reduction theorem (Theorem 3.17 in~\cite{Polyreduction}) is as follows:
\begin{theorem}\label{thm:polyred} Under the same conditions (C1) and (C2) above, let $\mu=(\mu_1,\dots,\mu_k)$ be a regular value of $J$ and assume that $G_\mu$ acts freely and properly on $J^{-1}(\mu)$. Then the reduced space $ J^{-1}(\mu)/G_\mu$ admits a unique polysymplectic structure $(\omega^1_\mu,\dots,\omega^k_\mu)$ satisfying $\pi_\mu^*\omega^a_\mu=i_\mu^*\omega^a$, where $\pi_\mu\colon J^{-1}(\mu)\to J^{-1}(\mu)/G_\mu$ is the canonical projection and $i_\mu\colon J^{-1}(\mu)\to N$ is the canonical inclusion.
\end{theorem}

There is also a direct dynamical consequence of Theorem~\ref{thm:polyred} if we are given an invariant Hamiltonian and consider invariant solutions. 

\begin{definition}\label{def:invariantvf} A $k$-vector field is \emph{$G$-invariant} if satisfies:
\[
\Phi_g^{T^1_kP}\circ \pmb{X}=\pmb{X}\circ\Phi_g. 
\]
\end{definition}
In other words, each of the components $X_a$ are $G$-invariant vector fields in the usual sense, namely $T\Phi_g\circ X_a=X_a\circ\Phi_g$. When the group is clear from the context, we will simply say ``invariant'' $k$-vector field. The following is proved in~\cite{Polyreduction} (note that the result is stated here in a slightly different form):

\begin{theorem}\label{thm:polyred2} Under the same conditions of Theorem~\ref{thm:polyred}, let $H\colon N\to\R$ be a $G$-invariant Hamiltonian and denote by $H_\mu$ its reduction to $N_\mu\equiv J^{-1}(\mu)/G_\mu$. Let $\pmb{X}=(X_1,\dots,X_k)$ be a solution of~\eqref{eq:k-Sym} with Hamiltonian $H$ and assume that:
\begin{enumerate}[label={(\roman*})]
 \item $\pmb{X}$ is $G_\mu$-invariant.
 \item The restriction $X_a\mid_{J^{-1}(\mu)}$ is tangent to $J^{-1}(\mu)$.
\end{enumerate}
Then the projection $\overline{\pmb{X}}_{\mu}$ of $\pmb{X}\mid_{J^{-1}(\mu)}$ on $N_\mu$ is a solution of~\eqref{eq:k-Sym} with Hamiltonian $H_\mu$. 
\end{theorem}

We will write 
\[
\pmb{X}_{\mu}\equiv \pmb{X}\mid_{J^{-1}(\mu)}. 
\]
A $k$-vector field $\pmb{X}$ such that $X_a\mid_{J^{-1}(\mu)}$ is tangent to $J^{-1}(\mu)$ (as in condition $(ii)$ of Theorem~\ref{thm:polyred2}) is said to be tangent to $J^{-1}(\mu)$. This means that $\pmb{X}_{\mu}$ can be considered as a $k$-vector field on $J^{-1}(\mu)$. For such a $k$-vector field, one identifies integral sections of $\pmb{X}$ through points in $J^{-1}(\mu)$ with integral sections of $\pmb{X}\mid_{J^{-1}(\mu)}$. 

\begin{remark}\label{remark1}  Theorem 4.4 in \cite{Polyreduction} (which corresponds to Theorem~\ref{thm:polyred2} above) requires the stronger condition that $\pmb{X}$ is $G$-invariant (rather than only $G_\mu$-invariant). But, reading through the proof of Theorem 4.4, it is clear that $G_\mu$-invariance suffices to reduce the $k$-vector field $\pmb{X}$. 
\end{remark}

We draw the attention of the reader to the fact that, in the present context, Noether's Theorem reads \cite{Symmetriesk}
\begin{equation}\label{eq:noetherfields}
\sum_a X_a (J^a_\xi)=0 
\end{equation}
for each $\xi\in\lag$ (strictly speaking, the version of Noether's Theorem discussed in \cite{Symmetriesk} applies to the case of the cotangent bundle of $k^1$-covelocities, but the result can be easily extended to a general polysymplectic manifold). This is the divergence property that we referred to in the introduction. Condition $(ii)$ in Theorem~\ref{thm:polyred2} is much more restrictive than~\eqref{eq:noetherfields}: it requires, in particular, that each of the terms $X_a(J^b_\xi)$ vanish \emph{separately}. This means, in particular, that only very specific solutions $\pmb{X}$ can be reduced to solutions of the reduced problems. However, it has been show in~\cite{Polyreduction} that a large class of examples (constructed from an invariant metric on $\lag$) contain solutions which fit within this category.

\subsection{Reconstruction} \label{recsec}

We will now discuss a method to lift an integral section of the reduced $k$-vector field $\overline{\pmb{X}}_{\mu}$ to an integral section of the original vector field $\pmb{X}$, when possible. 

To simplify the notation, we will describe the situation as follows. We assume that a manifold $P$ has an action $\Phi_g\colon P\to P$ of a Lie group $G$, and we write $\pi_P\colon P\to P/G$ for the corresponding principal bundle. At the end of this section we will take $P=J^{-1}(\mu)$ and $G=G_\mu$ to relate the results to the polysymplectic reduction case. Since the basic setting here is a principal bundle $P\to P/G$, which is precisely that of the Lagrange-Poincaré reduction, many of our results can be found independently in~\cite{LTM_LP}.

We consider an invariant $k$-vector field $\pmb{X}$ on $P$. It  defines, uniquely, a reduced $k$-vector field $\overline{\pmb{X}}$ on $P/G$ by the following relation:
\[
T\pi_P\circ X_a=\overline{X}_a\circ\pi_P.
\]

We first address the relation between the integrability of the (invariant) $k$-vector field $\pmb{X}$ and of the corresponding reduced vector field $\overline{\pmb{X}}$. It is clear that the integrability of $\pmb{X}$ implies that of $\overline{\pmb{X}}$: indeed, it suffices to observe that, in view of the integrability of $\pmb{X}$, we have
\[
[\overline{X}_a,\overline{X}_b]=[T\pi_P(X_a),T\pi_P(X_b)]=T\pi_P\left([X_a,X_b]\right)=0.
\]

We now seek to find conditions for the converse. For simplicity and concreteness in the proofs, one may assume that a principal connection $\Ac$ on the bundle $\pi_P\colon P\to P/G$ has been chosen. While this is not necessary yet (see~\cite{LTM_LP}), we will anyhow have to pick a principal connection on $P\to P/G$ later for the effective reconstruction of solutions. Each of the vector fields $X_a$ decomposes then as 
\[
X_a={\rm Hor}(X_a)+{\rm Ver}(X_a) = \overline{X}_a^h+{\rm Ver}(X_a),
\]
where $(\cdot)^h\colon \pi_P^*T(P/G)\to TP$ is the horizontal lift (we will omit the point where the lift occurs when is clear from the context). Let us assume that $\overline{\pmb{X}}$ is integrable. The bracket $[X_a,X_b]$ is decomposed as follows:
\begin{align}\label{eq:bracket}
[X_a,X_b]&={\rm Hor}([X_a,X_b])+ {\rm Ver}([X_a,X_b])= [\overline{X}_a,\overline{X}_b]^h+{\rm Ver}([X_a,X_b])\nonumber \\
&={\rm Ver}([X_a,X_b]).
\end{align}
We observe that $[X_a,X_b]$ is vertical (this condition does not depend on the chosen connection). In other words, we have the following:
\begin{lemma}\label{lem:integrability}
Assume that $\overline{\pmb{X}}$ is integrable. Then $\pmb{X}$ is integrable if and only if the vertical part of the Lie brackets $[X_a,X_b]$ vanishes. 
\end{lemma}

Given an integral section $\overline{\phi}\colon\R^k\to P/G$ of $\overline{\pmb{X}}$ at $[p]=\pi_P(p)$, we can construct the pull-back bundle $\overline{\phi}^*P$: 

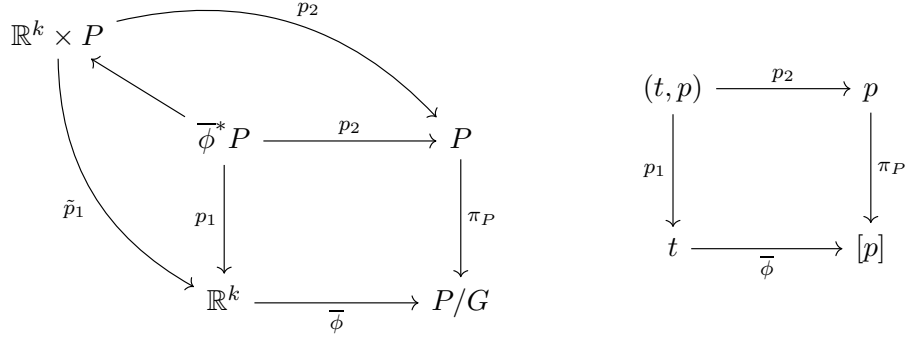
\begin{equation*}
\begin{tikzcd}
\R^k\times P\arrow[ddr,"\tilde p_1"',bend right=30]\arrow[drr,"\tilde p_2",bend left=30]&&[3em] \\
&\overline{\phi}^*P\arrow[r,"p_2"]\arrow[d,"p_1"']\arrow[ul]& P\arrow[d,"\pi_P"] \\[2em]
&\R^k \arrow[r, "\overline{\phi}"' ]& P/G
\end{tikzcd}\hspace{2em}  
\begin{tikzcd}
&(t,p)\arrow[r,"p_2"]\arrow[d,"p_1"']&[2em] p\arrow[d,"\pi_P"] \\[2em]
&t \arrow[r, "\overline{\phi}"' ]& {[p]}
\end{tikzcd}
\end{equation*}
Recall that $\overline{\phi}^*P\subset \R^k\times P$ is the submanifold
\[
\overline{\phi}^*P=\{(t,p)\in \R^k\times P\st \overline{\phi}(t)=\pi_P(p)\}. 
\]
It is a $G$-bundle over $\R^k$ with action $g\cdot (t,p)=(t,g\cdot p)$. A tangent vector in $T(\overline{\phi}^*P)$ at $(t,p)$ is  a pair $(v_t,v_p)$ with $T\overline{\phi}(v_t)=T\pi_P(v_p)$. The $k$-vector field $\pmb{X}$ defines a distribution $\mathcal{D}_{\pmb{X}}\subset TP$ of dimension $k$:
\[
\mathcal{D}_{\pmb{X}}=\langle X_1,\dots,X_k\rangle. 
\]
We consider the distribution $\mathcal{H}(\pmb{X},\overline{\phi})$ in $\overline{\phi}^*P$ obtained as the inverse image via $p_2$ of $\mathcal{D}_{\pmb{X}}$:
\[
\mathcal{H}(\pmb{X},\overline{\phi})=(Tp_2)^{-1}(\mathcal{D}_{\pmb{X}}),
\]
i.e. at a point $(t,p)\in \overline{\phi}^*P$ we have
\[
\left.\mathcal{H}(\pmb{X},\overline{\phi})\right|_{(t,p)}=(T_{(t,p)}p_2)^{-1}\big(\left.\mathcal{D}_{\pmb{X}}\right|_{p}\big).
\]
Let $Z_1=(V_1,X_1)$ be a vector field on $\overline{\phi}^*P\subset \R^k\times P$ such that $Tp_2(Z_1)=X_1$. Then $T\overline{\phi}(V_1)=T\pi_P(X_1)=\overline{X}_1$, and therefore $V_1=\partial/\partial t^1$. Similarly one finds $Z_2,\dots,Z_k$. Thus:
\begin{equation*}
\mathcal{H}(\pmb{X},\overline{\phi}) =\langle Z_1,\dots,Z_k\rangle
=\left\langle \fpd{}{t^1}+X_1,\dots,\fpd{}{t^k}+X_k\right\rangle.
\end{equation*}
Note that each of the vector fields 
\[
\left.Z_a\right|_{(t,p)}=\left.\fpd{}{t^a}\right|_t+\left.X_a\right|_p
\]
at a point $(t,p)\in \overline{\phi}^*P$ is interpreted as a vector field on $\overline{\phi}^*P$ as follows: $Z_a$ defines a vector field on $\R^k\times P$ tangent to $\overline{\phi}^*P$, and therefore its restriction to $\overline{\phi}^*P$ defines a vector field on $\overline{\phi}^*P$. We will not make a notational distinction between $Z_a\in \mathfrak{X}(\R^k\times P)$ and its restriction $Z_a\in \mathfrak{X}(\overline{\phi}^*P)$.

\begin{lemma} The distribution $\mathcal{H}(\pmb{X},\overline{\phi})$ defines a principal connection on $p_1\colon \overline{\phi}^*P\to\R^k$.
\end{lemma}
\begin{proof} Clearly, $\mathcal{H}(\pmb{X},\overline{\phi})$ defines a distribution of dimension $k$ in $\overline{\phi}^*P$ which is complementary to the vertical distribution of the bundle $p_1$. It is moreover an invariant distribution since each of the $X_a$ is.
\end{proof}

The horizontal and vertical parts of a vector field $Y\in\mathfrak{X}(\overline{\phi}^*P)$ w.r.t.\ the previous connection will be denoted by
\[
Y= {\rm Ver}_{\mathcal{H}(\pmb{X},\overline{\phi})}(Y)+{\rm Hor}_{\mathcal{H}(\pmb{X},\overline{\phi})}(Y)
\]
to distinguish them from those associated to $\Ac$. Take $0\in\mathfrak{X}(\R^k)$ and $X\in\mathfrak{X}(P)$ vector fields with $0+X\in T(\overline{\phi}^*P)$, then $T\pi_P(X)=0$ or, in other words, $X={\rm Ver}(X)$. Since vertical vector fields (w.r.t.\ $p_1$) on $T(\overline{\phi}^*P)$ are precisely of the form $0+X$,  a vertical vector in $\overline{\phi}^*P$ can be though of as a vertical vector on $P$. We then have:

\begin{proposition}\label{pro:reconstruction} Let $\pmb{X}$ be an invariant $k$-vector field on $P$. Then $\pmb{X}$ is integrable if, and only if the following two conditions are satisfied: 
\begin{enumerate}[label={(\roman*})]
 \item The reduced vector field $\overline{\pmb{X}}$ is integrable.
 \item For each integral section $\overline{\phi}\colon\R^k\to P/G$ of $\overline{\pmb{X}}$ the connection $\mathcal{H}(\pmb{X},\overline{\phi})$ is flat.
\end{enumerate}
\end{proposition}
\begin{proof} If $Z_a=\partial/\partial t^a+ X_a$ and $Z_b=\partial/\partial t^b+ X_b$ are two horizontal vector fields on $\overline{\phi}^*P$,
their bracket reads
\[
[Z_a,Z_b] =0+ [X_a,X_b]
\]
and the curvature is
\[
K_{\pmb{X},\overline{\phi}}(Z_a,Z_b)=- {\rm Ver}_{\mathcal{H}(\pmb{X},\overline{\phi})}(0+ [X_a,X_b])=-{\rm Ver}([X_a,X_b]).
\]
Note that in the preceding expression we have abused slightly the notation to identify the vertical subbundles of $T(\overline{\phi}^*P)$ and $TP$. Thus, if the curvature vanishes for each pair $Z_a,Z_b$, one obtains precisely the condition for the integrability of  $\pmb{X}$ under the assumption that $\overline{\pmb{X}}$ is integrable, see~\eqref{eq:bracket}. We conclude that if $\pmb{X}$ is integrable then both (i) and (ii) hold, and conversely.
\end{proof}

The same result can be found in \S 3.2 of~\cite{LTM_LP}, where the connection on $\overline{\phi}^*P$ is defined in an alternative way using the connection associated to $\pmb{X}$. Note that one can also obtain this connection as the restriction of   
the distribution $\mathcal{H}$ defined by
\begin{equation*}
\mathcal{H}(\pmb{X})=\left\langle \fpd{}{t^1}+ X_1^i\fpd{}{q^i},\dots, \fpd{}{t^k}+ X_k^i\fpd{}{q^i}\right\rangle\subset T(\R^k\times Q)  
\end{equation*}
to the submanifold $\overline{\phi}^*P\subset \R^k\times P$. 

The proof of Propositon~\ref{pro:reconstruction} shows the following: if $[p]\in P/G$ is a point and we take an integral section  $\overline{\phi}$ of $\overline{\pmb{X}}$ through $[p]$, then given any point $p\in P$ with $\pi_P(p)=[p]$ there exists an integral section $\phi$ of $\pmb{X}$ through $p\in P$. It also clear that $\pi_P\circ\phi=\overline{\phi}$ since
\[
T(\pi_P\circ\phi)(\partial/ \partial t^a)=T\pi\circ X_a=\overline{X}_a=T\overline{\phi} (\partial/ \partial t^a),
\]
and both $\pi_P\circ\phi$ and $\overline{\phi}$ agree on $0\in\R^k$ (with value $[p]$). The proof, however,  does not give a constructive procedure to find such $\phi$.  This procedure, namely the effective reconstruction of such a $\phi$ starting from a reduced integral section $\overline{\phi}$,  has been described already in~\cite{LTM_LP} in the polysymplectic formalism, and we will only recollect here the main results without proofs. 

We will denote by $\overline{\pmb{X}}^h$ the horizontal lift of $\pmb{X}$ w.r.t.\ $\pmb{\Ac}$, which is by definition the $k$-vector field on $P$ with components
\[
\overline{\pmb{X}}^h=(\overline{X}_1^h,\dots,\overline{X}_k^h), 
\]
where the notation $(\cdot)^h$ for the vector fields in the brackets is the usual horizontal lift w.r.t. $\Ac$. By construction, $\overline{\pmb{X}}^h$ projects onto $\overline{\pmb{X}}$ and we can write just like in~\eqref{eq:bracket} 
\[
[\overline{X}_a^h,\overline{X}_b^h]=[\overline{X}_a,\overline{X}_b]^h-K(X_a,X_b),
\]
where we have used that ${\rm hor}(X_a)=\overline{X}_a^h$. Clearly, $\overline{\pmb{X}}^h$ is integrable iff both of the following conditions are satisfied: (i) $\overline{\pmb{X}}$ is integrable, and (ii) the curvature $K(X_a,X_b)$ vanishes for each $a,b=1,\dots k$. 

\begin{definition} An integral section 
\[
\overline{\phi}_h\colon \R^k\to P 
\]
of $\overline{\pmb{X}}^h$ is a \emph{horizontal lift of $\overline{\phi}$} if $\pi_P\circ \overline{\phi}_h=\overline{\phi}$.

\end{definition}
When $\overline{\pmb{X}}^h$ is integrable such an horizontal lift always exists for the same reasons that $\phi$ does (see the argument after the proof of Propositon~\ref{pro:reconstruction}). To determine the desired integral section $\phi\colon \R^k\to P$ of $\pmb{X}$ one looks for the map $g\colon \R^k\to G$ such that  
\begin{equation}\label{eq:phase}
\phi(t)=g(t)\cdot \overline{\phi}_h(t). 
\end{equation}
If one uses~\eqref{eq:phase} in the defining relation for the integral section $\phi$, i.e. 
\[
\phi^{(1)}=\pmb{X}\circ\phi, 
\]
one arrives at the following \emph{reconstruction equation}~\cite{LTM_LP}:
\begin{equation}\label{eq:reconstruction}
g^{-1}\cdot g^{(1)}=\pmb{\Ac}(\pmb{X}\circ \overline{\phi}_h).
\end{equation}
The previous equation compares directly to the well-known case of mechanics. Indeed, when $k=1$, the expression~\eqref{eq:reconstruction} reduces to 
\begin{equation*}
g^{-1}\cdot \dot g=\Ac(X_H\circ \overline{\gamma}_h), 
\end{equation*}
which is the reconstruction equation in the case of reduction in Hamiltonian mechanics (see for example~\cite{phases} for details in the context of symplectic reduction). Here $X_H$ the Hamiltonian vector field and $\overline{\gamma}_h$ is the horizontal lift of the reduced solution $\overline{\gamma}$.   

When Proposition~\ref{pro:reconstruction} is applied to the case $P=J^{-1}(\mu)$ of interest, we have the following:

\begin{theorem}\label{thm:reconstruction} Under the same conditions as Theorem~\ref{thm:polyred2}, let $\overline{\phi}_\mu\colon \R^k\to N_\mu$ be an integral section of $\overline{\pmb{X}}_{\mu}$ such that the connection $\mathcal{H}(\pmb{X}_\mu,\overline{\phi}_\mu)$ is flat. Then there exists an integral section $\phi_\mu\colon \R^k\to J^{-1}(\mu)$ of $\overline{\pmb{X}}_{\mu}$ with $\pi_\mu\circ\phi_\mu=\overline{\phi}_\mu$.

Moreover, if for a given principal connection $\Ac$ the horizontal lift $\overline{\pmb{X}}_{\mu}^h$ is integrable, then such an integral section $\phi_\mu$ can be computed as 
\[
\phi_{\mu}(t)=g(t)\cdot(\overline{\phi}_\mu)_h(t), 
\]
where $g(t)$ satisfies the reconstruction equation~\eqref{eq:reconstruction}.
\end{theorem}
We point out that in Theorem~\ref{thm:reconstruction} the connection $\mathcal{H}(\pmb{X}_\mu,\overline{\phi}_\mu)$ is a $G_\mu$ principal connection on the bundle $\overline{\phi}_\mu^*(J^{-1}(\mu))$. 

\begin{remark} We will not obtain the general coordinate expressions for the reconstruction equations. In the Lagrangian case (Section \ref{sec:Routh}) such expressions can be adapted from those in~\cite{LTM_LP}.
\end{remark}

\section{The reduction of the cotangent bundle of \texorpdfstring{$k^1$}{k1}-covelocities}\label{sec:cotangentreduction}

We will now discuss in detail how the polysymplectic reduction theorem applies to the particular case of the cotangent bundle of $k^1$-covelocities $(T^1_k)^*Q$. The field equations in the $k$-symplectic Hamiltonian field theory look for integral sections of some Hamiltonian $k$-vector field defined on $(T^1_k)^*Q$, so this reduction is important in its own. Besides, we will build on these results to discuss Routh reduction later in Section~\ref{sec:Routh}.   

Our starting point is a free and proper action $\Phi_g\colon Q\to Q$ of $G$ on $Q$ and a Hamiltonian $H\colon (T^1_k)^*Q\to\R$ which is invariant under the canonical prolongation of the action to $(T^1_k)^*Q$. It is well-known~\cite{Gunter} that in this case an equivariant momentum map is given by the following family of maps $J^a\colon (T_k^1)^*Q\to \lag^*$:
\begin{equation}\label{eq:liftedmmap}
\langle J^a(\bm{\alpha}_q),\xi\rangle=\langle\alpha^a_q, \xi_{Q}(q)\rangle ,\quad \bm{\alpha}_q=(\alpha_q^1,\dots, \alpha_q^k),\;\text{for all } \xi\in\lag.
\end{equation}
It is shown in~\cite{Polyreduction} that this example is amenable to polysymplectic reduction, i.e. that both (C1) and (C2) in Theorem~\ref{thm:polyred} hold under the assumption of a free action (actually under the weaker condition of an infinitesimally free action). 

We need some notations before moving on. Given a principal connection $\Ac$ on $Q\to Q/G$ and a value $\nu\in\lag^*$, it is possible to define 1-form $\Ac_\nu$ on $Q$ as follows:
\[
\Ac_\nu(v_q)=\langle\nu,\Ac(v_q)\rangle.  
\]
The 2-form $d\Ac_\nu$ reduces to a 2-form on $Q/G_\nu$, with $G_\nu$ the isotropy group of $\nu\in\lag$ under the coadjoint action. This follows easily from the $G_\nu$-equivariance of $\Ac_\nu(v_q)$ (cf. \cite{RouthMarsden} for further details). We will denote this reduced form on $Q/G_\nu$ by $B_\nu$; in reduction terminology, this reduced form is often called the ``magnetic term''.

For a given regular value $\mu=(\mu_1,\dots,\mu_k)\in\lagdk$, we consider the family of 2-forms on $Q$ given by $d\Ac_{\mu_a}$. Since $G_\mu\subset G_{\mu_a}$, each of them drops to a 2-form on $Q/G_\mu$ that we will denote $\Bc_{\mu_a}$. 

\begin{theorem}\label{thm:identification} Any choice of a principal connection $\Ac$ on $Q\to Q/G$ gives a  polysymplectomorphism 
\begin{equation}\label{eq:identification}
\left((T^1_k)^*Q\right)_\mu \simeq Q/G_\mu \times_{Q/G} \big(\smallunderbrace{T^*(Q/G)\oplus\;\dots\oplus T^*(Q/G)\;}_{k\; {\rm copies}}\big),
\end{equation}
where the space on the right-hand side (RHS) is endowed with the polysyplectic structure
\begin{equation*}
\omega_\mu^a= ({\rm pr}^a_2)^*\omega_{Q/G}-({\rm pr}_1)^*\Bc_{\mu_a}. 
\end{equation*}
\end{theorem}
The Whitney sum in~\eqref{eq:identification} is over $Q/G$. The notation ${\rm pr}^a_2$ stands for the projection of the $a$-th component of the RHS of~\eqref{eq:identification} onto $T^*(Q/G)$, and similarly ${\rm pr}_1$ is the projection onto $Q/G_\mu$. If we write $p_\mu\colon Q/G_\mu\to Q/G$ for the projection $[q]_\mu\mapsto [q]$, the space on the RHS of~\eqref{eq:identification} is a pullback bundle:
\begin{equation*}
\begin{tikzcd}
p_\mu^*\left((T^1_k)^*(Q/G)\right)\arrow[r,dashed]\arrow[d,dashed]& (T^1_k)^*(Q/G)\arrow[d] \\Q/G_\mu\arrow[r, "p_\mu"' ]& Q/G
\end{tikzcd} 
\end{equation*}

\begin{proof} We extend and adapt the proof of the cotangent bundle reduction theorem (see~\cite{stagesbook} and references therein) to the polysymplectic setting.

\textsc{Step 1 (The level set of zero):} Asume for a moment that $\mu=\pmb{0}\equiv (0,\dots,0)$. Then  from the definition of the momentum map~\eqref{eq:liftedmmap}, we have an identification
\[
(J^a)^{-1}(0)=T^*Q\oplus\dots\oplus \underbrace{(V\pi)^\circ}_{\text{$a$-th comp.}} \oplus\dots\oplus T^*Q \subset (T^1_k)^*Q 
\]
where $V\pi\subset TQ$ is the vertical subbundle w.r.t.\ the projection $\pi\colon Q\to Q/G$, and $(V\pi)^\circ\subset T^*Q$ denotes its annihilator. Thus:
\[
J^{-1}(\pmb{0})=\cap_a \left[(J^a)^{-1}(0)\right] = (V\pi)^\circ\oplus\dots\oplus (V\pi)^\circ\subset (T^1_k)^*Q.
\]
There is a canonical identification between $(V\pi)^\circ$ and $T^*(Q/G)\times_{Q/G}Q$
\begin{align*}
(V\pi)^\circ &\to Q\times_{Q/G}T^*(Q/G), \\
\alpha_q &\mapsto (q,\tilde\alpha_{[q]}),  
\end{align*}
where $\tilde\alpha_{[q]}$ is uniquely defined by the relation $\langle\tilde\alpha_{[q]},\tilde v_{[q]}\rangle = \langle\alpha_q,v_q\rangle$ for any $v_q$ with $T\pi(v_q)=\tilde v_{[q]}$ (this is well-defined). Therefore, we also have an identification 
\[
\mathcal{T}\colon J^{-1}(\pmb{0}) \to Q\times_{Q/G}\left(T^*(Q/G)\oplus\dots\oplus T^*(Q/G)\right) =\pi^*\left((T^1_k)^*(Q/G)\right).
\]
Moreover, under the previous identification, $G$ acts on $\pi^*\left((T^1_k)^*(Q/G)\right)$ as:
\[
g\cdot (q,\alpha^1_{[q]},\dots,\alpha^k_{[q]})= (g\cdot q,\alpha^1_{[q]},\dots,\alpha^k_{[q]}).
\]

\textsc{Step 2 (The ``momentum shift''):} Now let $\mu$ be arbitrary. For each $a=1,\dots,k$, we consider the 1-form  $\Ac_{\mu_a}$ on $Q$ . Define the map (which is the generalization of the ``momentum shift'' to the polysymplectic setting):
\begin{align*}
\mathcal{S}\colon (T^1_k)^*Q&\to (T^1_k)^*Q, \\
(\alpha^1_q,\dots,\alpha^k_q)&\mapsto (\alpha^1_q-(\Ac_{\mu_1})_q,\dots,\alpha^k_q-(\Ac_{\mu_k})_q).
\end{align*}
Note that $\mathcal{S}$ depends on both $\mu$ and the chosen principal connection $\Ac$. The map $\mathcal{S}$ is such that $\mathcal{S}^*\theta^a_Q=\theta^a_Q-(\pi^a_Q)^*\Ac_{\mu_a}$. To see this, note that if we write $(v_1,\dots,v_k)\in T\big((T^1_k)^*Q\big)$ we have
\begin{align*}
\mathcal{S}^*(\theta^a_Q)_{(\alpha^1_q,\dots,\alpha^n_q)}(v_1,\dots,v_k)&=(\theta_Q)_{\alpha^a_q-\Ac_{\mu_a}}(T\pi^a_Q(v_a))\\
&=\big( (\theta^a_Q)_{(\alpha^1_q,\dots,\alpha^n_q)}-(\pi^a_Q)^*\Ac_{\mu_a}\big)(v_1,\dots,v_k), 
\end{align*}
where we have used the definition of $\theta^a_Q$. In particular, taking the exterior derivative and using $\omega^a_Q=-d\theta^a_Q$ we find
\begin{equation}\label{eq:shift}
\mathcal{S}^* \omega^a_Q= \omega^a_Q+(\pi^a_Q)^*d\Ac_{\mu_a}.
\end{equation}

We consider now the restriction of $\mathcal{S}$ to $J^{-1}(\mu)$, which we denote by the same symbol. It maps, diffeomorphically, $J^{-1}(\mu)$ onto $J^{-1}(\pmb{0})$ (``shifts the momentum''):
\[
\mathcal{S}\colon  J^{-1}(\mu)\to J^{-1}(\pmb{0}).
\]
Indeed, for each $a$ we have:
\[
(J^a\circ \mathcal{S})(\alpha^1_q,\dots,\alpha^n_q)=J^a(\alpha^1_q-\Ac_{\mu_1},\dots,\alpha^k_q-\Ac_{\mu_k})=\mu_a-\mu_a=0.
\]
Composing with $\mathcal{T}$ we get a diffeomorphism
\begin{align*}
\tau=\mathcal{T}\circ\mathcal{S}\colon J^{-1}(\mu)&\to \pi^*\left((T^1_k)^*(Q/G)\right)
\end{align*}
which is $G_\mu$-equivariant (this follows from the fact that each $\Ac_{\mu_a}$ is $G_\mu$-equivariant), and hence reduces to a diffeomorphism:
\[
\tau_\mu\colon J^{-1}(\mu)/G_\mu\to  p_\mu^*\left((T^1_k)^*(Q/G)\right).
\]
\textsc{Step 3 (The symplectic form):} Since $\tau_\mu$ is a diffeomorphism, it only remains to check that it relates both polysymplectic structures. If we let $\tilde{\omega}^a_\mu$ be the reduced polysymplectic structure on $J^{-1}(\mu)/G_\mu$ given by Theorem~\ref{thm:polyred}, then we must check that 
\[
\tilde{\omega}^a_\mu=\tau_\mu^*\omega^a_\mu=\tau_\mu^* \left[({\rm pr}^a_2)^*\omega_{Q/G}-({\rm pr}_1)^*\Bc_{\mu_a}\right].
\]
The situation is summarized in the following commutative diagram:
\begin{equation*}
\begin{tikzcd}
(T^1_k)^*Q\arrow[rr,"\mathcal{S}"] &  & (T^1_k)^*Q &\\
J^{-1}(\mu)\arrow[rr,"\tau"]\arrow[u,"i_\mu"]\arrow[dd,"\pi_\mu"']&& \pi^*\left((T^1_k)^*(Q/G)\right)\arrow[dd,"\pi^0_\mu"]\arrow[u,"j_0=(\mathcal{T}^{-1}\circ i_0)"'] \arrow[dr] &\\
&& & T^*(Q/G)\\
J^{-1}(\mu)/G_\mu \arrow[rr,"\tau_\mu"']&& p_\mu^*\left((T^1_k)^*(Q/G)\right) \arrow[ur,"{\rm pr}^a_2"']\arrow[r,"{\rm pr}_1"'] & Q/G_\mu
\end{tikzcd}
\end{equation*}
We will denote by $i_0\colon J^{-1}(0)\to (T^1_k)^*Q$ the inclusion at $\mu=\pmb{0}$ and write
\[
j_0=(\mathcal{T}^{-1}\circ i_0)\colon  \pi^*\left((T^1_k)^*(Q/G)\right) \to (T^1_k)^*Q.
\]
Also, we will write $\pi^0_\mu\colon \pi^*\left((T^1_k)^*(Q/G)\right)\to p_\mu^*\left((T^1_k)^*(Q/G)\right)$ for the quotient projection
\[
\pi^0_\mu([q],\bm{\alpha}_{[q]})= ([q]_\mu,\bm{\alpha}_{[q]}).
\]

Since $\tilde{\omega}^a_\mu$ is uniquely determined from the relation
$\pi_\mu^*\tilde{\omega^a_\mu}=i_\mu^*\omega_Q^a$, we only need to show that
\[
\pi_\mu^* \tau_\mu^* \big(({\rm pr}^a_2)^*\omega_{Q/G}-({\rm pr}_1)^*\Bc_{\mu_a}\big)=i_\mu^*\omega_Q^a,
\]
or
\begin{equation}\label{eq:equalityforms}
\tau^*(\pi^0_{\mu})^* \big(({\rm pr}^a_2)^*\omega_{Q/G}-({\rm pr}_1)^*\Bc_{\mu_a}\big)=i_\mu^*\omega_Q^a.
\end{equation}
From the definitions of the maps involved, it is not hard to check that the following relations hold:
\[
(\pi^0_{\mu})^* ({\rm pr}^a_2)^*\omega_{Q/G}=j_0^*\omega^a_Q,\qquad  (\pi^0_{\mu})^* ({\rm pr}_1)^*\Bc_{\mu_a} =j_0^*(\pi^a_Q)^*d\Ac_{\mu_a}.
\]
Therefore, the equality~\eqref{eq:equalityforms} will hold provided 
\[
\tau^*j_0^*(\omega^a_Q-(\pi^a_Q)^*d\Ac_{\mu_a})= i_\mu^*\omega_Q^a.
\]
But this follows if we rewrite~\eqref{eq:shift} in the form  
\[
\omega^a_Q=\mathcal{S}^* \omega^a_Q-(\pi^a_Q)^*d\Ac_{\mu_a}=\mathcal{S}^*(\omega^a_Q-(\pi^a_Q)^*d\Ac_{\mu_a}), 
\]
and now:
\[
i_\mu^*\omega_Q^a=i_\mu^*\mathcal{S}^*(\omega^a_Q-(\pi^a_Q)^*d\Ac_{\mu_a})=\tau^*j_0^*(\omega^a_Q-(\pi^a_Q)^*d\Ac_{\mu_a})=\tau^*(\pi^0_{\mu})^* \big(({\rm pr}^a_2)^*\omega_{Q/G}-({\rm pr}_1)^*\Bc_{\mu_a}\big).
\]
\end{proof}

\section{Routh reduction in the polysymplectic formalism}\label{sec:Routh}

In this section we obtain a polysymplectic version of Routh reduction. To do that, we will follow~\cite{quasi} and adapt the techniques to the polysymplectic setting. In a nutshell, the methodology consists on applying the polysymplectic reduction Theorem~\ref{thm:polyred} to the given regular Lagrangian system on $T^1_kQ$.

\subsection{The reduction of the tangent bundle of \texorpdfstring{$k^1$}{k1}-velocities}

Building on the results of the previous section, we now discuss the reduction of a Lagrangian field theory defined on $T^1_kQ$. The assumption is that we have a Lagrangian $L\colon T^1_kQ\to \R$ which is hyperregular and $G$-regular (to be defined later, see Definition~\ref{def:Gregularity}), and a free and proper action $\Phi_g$ on $Q$ such that $L$ is invariant w.r.t. its canonical prolongation to 
$T^1_kQ$. The requirement of hyperregularity can be replaced for that of regularity with small changes. We will use the following result:

\begin{lemma}\label{lem:lema1} Let $(N,\omega^a)$ and $(N',\omega'^a)$ be polysymplectic manifolds and $f\colon N\to N'$ a polysymplectomorphism. Assume that $G$ acts canonically on both $N$ and $N'$ with momentum maps $J\colon N\to \lagdk$ and $J'\colon N'\to \lagdk$, and that $f$ is equivariant w.r.t. the $G$-actions and satisfies $f^*J'=J$. 

\begin{equation*}
\begin{tikzcd}
(N,\omega^a)\arrow[rr,"f"]\arrow[dr,"J"']&& (N',\omega'^a)\arrow[dl,"J'"] \\
& \lagdk  & 
\end{tikzcd}
\end{equation*}

\noindent Then, if $f_\mu\colon N_\mu\to N'_\mu$ denotes the map between the reduced polysymplectic spaces induced by $f$, $f_\mu$ is a polysymplectomorphism, i.e.
\[
f_\mu^*\omega'^a_\mu=\omega^a_\mu. 
\]
\end{lemma}

\begin{proof} The proof follows from the characterization of the reduced polysymplectic forms and is similar to that of Theorem~5 in~\cite{quasi} in the symplectic case.

First we observe that, since $f^*J'=J$, the map $f$ restricts to a diffeomorphism $f_r\colon J^{-1}(\mu)\to J'^{-1}(\mu)$ which is $G_\mu$-equivariant, and thus $f_\mu\colon N_\mu\to N'_\mu$ is a diffeomorphism:  
\begin{equation*}
\begin{tikzcd}[column sep=huge,row sep= large]
(N,\omega^a)\arrow[r,"f"]& (N',\omega'^a) \\
J^{-1}(\mu)\arrow[d,"\pi_\mu"']\arrow[u,"i_\mu"]\arrow[r,"f_r"]& J'^{-1}(\mu)\arrow[d,"\pi'_\mu"]\arrow[u,"i'_\mu   "'] \\
(N_\mu,\omega^a_\mu) \arrow[r,"f_\mu"] & (N'_\mu,\omega'^a_\mu)
\end{tikzcd} 
\end{equation*}
We need to show that $f_\mu^*\omega'^a_\mu=\omega^a_\mu$. Since $\pi_\mu$ and $\pi'_\mu$ are submersions, it suffices to check that:
\[
f_r^*(\pi'_\mu)^* \omega'^a_\mu=(\pi_\mu)^*\omega^a_\mu.
\]
But this is easily shown using the characterization of the reduced polysymplectic forms:
\begin{align*}
f_r^*(\pi'_\mu)^* \omega'^a_\mu&=f_r^*(i'_\mu)^* \omega'^a=(i_\mu)^*f^*\omega'^a\\
&=(i_\mu)^*\omega^a=(\pi_\mu)^*\omega^a_\mu.
\end{align*}
\end{proof}

Recall that a regular Lagrangian $L$ defines the diffeomorphism $\pmb{F} L\colon T^1_kQ\to (T^1_k)^*Q$, see~\eqref{eq:Legendre}. We first want to observe that, since $L$ is invariant, the Legendre transformation is equivariant:
\[
\F L(g\cdot \pmb{v}_q)=g\cdot \F L(\pmb{v}_q). 
\]
The proof, componentwise, is similar to the mechanical case which can be found for example in~\cite{Foundations}, Corollary 4.2.14. With this in mind, it is clear that the map
\[
J_L=J\circ \F L\colon T^1_kQ\to\lagdk 
\]
is an equivariant momentum map for the polysymplectic action of $G$ on $T^1_k Q$. Its components are the maps 
\[
 J_L^a=J^a\circ \F L\colon  T^1_kQ\to \lag^*,
\]
i.e.
\begin{equation*}
\langle J_L^a(\pmb{v}_q),\xi\rangle=\langle(F L)^a(\pmb{v}_q), \xi_{Q}(q)\rangle ,\quad \pmb{v}_q=(v_{1q},\dots,v_{kq}),\;\text{for all } \xi\in\lag.
\end{equation*}

We will need the following map: 
\begin{align*}
\mathfrak{J}_L\colon T^1_kQ\times \lagk&\to \lagdk,\\  
(\pmb{v}_q,\bm{\xi})&\mapsto \mathfrak{J}_L(\pmb{v}_q,\bm{\xi})=J_L(\pmb{v}_q+\bm{\xi}_Q(q)),
\end{align*}
where $\bm{\xi}=(\xi_1,\dots,\xi_k)\in\lagk$ and $\bm{\xi}_Q(q)\in T^1_kQ$ is given by
\[
\bm{\xi}_Q(q)=\left((\xi_1)_Q(q),\dots,(\xi_k)_Q(q)\right). 
\]

\begin{definition}\label{def:Gregularity} Let $L\colon T^1_kQ\to \R$ be a  Lagrangian. We say that $L$ is $G$-regular if, for each $\pmb{v}_q\in T^1_kQ$, the map $\mathfrak{J}_L(\pmb{v}_q,\cdot)\colon \lagk\to \lagdk$ is a diffeomorphism.
\end{definition}

This definition is the natural extension of the notion of a $G$-regular Lagrangian to the polysymplectic setting~\cite{Routhstages}. There are other equivalent definitions of $G$-regularity~\cite{quasi}, and one might as well work with the extension of those to the $k$-symplectic framework, but we find Definition~\ref{def:Gregularity} to be most practical. The importance of $G$-regularity comes from the following fact:
\begin{proposition}\label{prop:identificationJL} If $L$ is $G$-regular, then there is a diffeomorphism  
\[
J_L^{-1}(\mu)/G_\mu \simeq Q/G_\mu \times_{Q/G} \big(\smallunderbrace{T(Q/G)\oplus\;\dots\oplus T(Q/G)\;}_{k\; {\rm copies}}\big)
\]
over the identity on $Q/G_\mu$.
\end{proposition}

\begin{proof} We first show that there is a diffeomorphism $\tilde\tau$ between $J_L^{-1}(\mu)$ and $\pi^*\left(T^1_k(Q/G)\right)$. We observe that $J_L^{-1}(\mu)=\cap_a \left[(J^a_L)^{-1}(\mu)\right]$. We define $\tilde\tau$ as follows:
\begin{align*}
\tilde\tau \colon \cap_a \left[(J^a_L)^{-1}(\mu)\right]  &\to  \pi^*\left(T^1_k(Q/G)\right),\\
(v_{1q}\dots,v_{kq})&\mapsto (q,T\pi(v_{1q}),\dots, T\pi(v_{kq})).
\end{align*}
We will now find explicitly the inverse $\tilde\tau$. An element in $\pi^*\left(T^1_k(Q/G)\right)$ is of the form
\[
(q,v_{1[q]},\dots,v_{k[q]})=\left(q,T\pi(v_{1q}),\dots,T\pi(v_{kq})\right)\in \pi^*\left(T^1_k(Q/G)\right) 
\]
for some $v_{1q},\dots,v_{kq}\in TQ$. Then we define the inverse of $\tilde\tau$ as:
\begin{align*}
(\tilde\tau)^{-1} \colon \pi^*\left((T^1_k)(Q/G)\right) &\to  \cap_a \left[(J^a_L)^{-1}(\mu)\right],\\
\left(q,T\pi(v_{1q}),\dots,T\pi(v_{kq})\right)&\mapsto \left(v_{1q}+(\xi_1)_Q(q),\dots,v_{kq}+(\xi_k)_Q(q)\right)=\pmb{v}_q+\bm{\xi}_Q(q),
\end{align*}
where $\xi_a\in\lag$ can be chosen so that $J_L(\pmb{v}_q+\bm{\xi}_Q(q))=\mu$ because of the assumed  $G$-regularity of $L$. We note that, if we had made a different choice for the vectors $v_{1q},\dots,v_{kq}$, say $w_{1q},\dots,w_{kq}$, then we would have $\pmb{v}_q=\pmb{w}_q+\bm{\eta}_Q(q)$ for some $\bm{\eta}\in\lagk$ and the result would be the same due to $G$-regularity. This means that $(\tilde\tau)^{-1}$ is well-defined.

The $G_\mu$-action on space $J_L^{-1}(\mu)$ is pushed forward to the space $\pi^*\left(T^1_k(Q/G)\right)$, where it reads:
\[
g\cdot(q,v_{1[q]},\dots,v_{k[q]})= (g\cdot q,v_{1[q]},\dots,v_{k[q]}).
\]
Therefore, the map $\tilde\tau$ drops to a diffeomorphism  
\[
\tilde\tau_\mu\colon J_L^{-1}(\mu)/G_\mu\to \pi^*\left(T^1_k(Q/G)\right)/G_\mu= p_\mu^*\left(T^1_k(Q/G)\right),
\]
and this gives the desired identification. 
\end{proof}

The conclusion is that we have the following diagram:
\begin{equation}\label{dia:red}
\begin{tikzcd}
\pi^*\left(T^1_k(Q/G)\right)\arrow[d,"\tilde\pi^0_\mu"']&J_L^{-1}(\mu)\arrow[l,"\tilde\tau"']\arrow[r,"\F L"]\arrow[d,"\tilde\pi_\mu"']& J^{-1}(\mu)\arrow[d,"\pi_\mu"]\arrow[r,"\tau"] & \pi^*\left((T^1_k)^*(Q/G)\right)\arrow[d,"\pi^0_\mu"]\\
p_\mu^*\left(T^1_k(Q/G)\right)&J_L^{-1}(\mu)/G_\mu \arrow[l,"\tilde\tau_\mu"]\arrow[r,"(\F L)_\mu"']& J^{-1}(\mu)/G_\mu \arrow[r,"\tau_\mu"'] & p_\mu^*\left((T^1_k)^*(Q/G)\right)
\end{tikzcd}
\end{equation}
where $\tilde\pi^0_\mu$ and $\tilde\pi_\mu$ are quotient maps. The map $(\F L)_\mu$, obtained from the equivariant diffeomorphism $\F L$, is then a polysymplectomorphism (Lemma~\ref{lem:lema1}). We point out that, as mentioned earlier, we also denote by $\F L$ the restriction of $\F L$ to $J^{-1}(\mu)$. 

\subsection{The Routhian}

Using the momentum shift $\mathcal{S}$ we write the first row of Diagram~\eqref{dia:red} as:
\[
\begin{tikzcd}
\pi^*\left(T^1_k(Q/G)\right)&J_L^{-1}(\mu)\arrow[l,"\tilde\tau"']\arrow[r,"\F L"]& J^{-1}(\mu)\arrow[r,"\mathcal{S}"] & J^{-1}(\pmb{0})\arrow[r,"\mathcal{T}"] & \pi^*\left((T^1_k)^*(Q/G)\right).
\end{tikzcd}
\]
The composition $\mathcal{S}\circ \F L\colon J_L^{-1}(\mu)\to J^{-1}(0)$ can be described as follows:
\begin{align*}
\langle\mathcal{S}\circ (F L)^a(\pmb{v}),w_q\rangle &= \left.\frac{d}{ds}\right|_{s=0}L(v_{1q},\dots,v_{aq}+sw_q,\dots,v_{kq})-\langle\Ac_{\mu_a},w_q \rangle\\
&=\left.\frac{d}{ds}\right|_{s=0}(L-(\pi_Q^a)^*\Ac_{\mu_a})(v_{1q},\dots,v_{aq}+sw_q,\dots,v_{kq}).
\end{align*}
This means that if we define the function $R\colon T^1_kQ\to\R$
\begin{align*}
R(\pmb{v})&=L(\pmb{v})-(\pi_Q^1)^*\Ac_{\mu_1}(\pmb{v})-\dots- (\pi_Q^k)^*\Ac_{\mu_k}(\pmb{v})\nonumber \\
&=L(\pmb{v})-\Ac_{\mu_1}(v_{1q})-\dots- \Ac_{\mu_k}(v_{kq}),  
\end{align*}
then its restriction to $J_L^{-1}(\mu)$ satisfies $\F R=\mathcal{S}\circ \F L$. We will now denote by
\[
\Ro\colon \pi^*\left(T^1_k(Q/G)\right)\to\R
\]
the function induced by $R$ on $\pi^*\left(T^1_k(Q/G)\right)$, and by
\[
\Ro_\mu\colon p_\mu^*\left(T^1_k(Q/G)\right)\to\R
\]
its reduction (the function induced by $\Ro$ on the quotient by $G_\mu$).

\begin{definition}
The function $\Ro_\mu$ (but also of $R$ and $\Ro$) is called the \emph{Routhian}.
\end{definition}
We recall that the meaning of the fiber derivatives $\F\Ro$ and $\F\Ro_\mu$ has been described at the end of Section~\ref{sec:HLfieldtheory}. 
\begin{lemma} The following holds:
\begin{enumerate}[label={(\roman*})]
\item $\Ro$ is $G_\mu$-invariant.
\item $\F\Ro=\tau\circ \F L\circ(\tilde\tau)^{-1}$.
\item $\F\Ro_\mu=\tau_\mu \circ (\F L)_\mu\circ (\tilde\tau_\mu)^{-1}$. 
\end{enumerate}
\end{lemma}

\begin{proof} For each $a$, the term $\Ac_{\mu_a}$ in $\Ro$ is $G_\mu$-invariant (see Section~\ref{subsec:connections}). Since by assumption $L$ is $G_\mu$-invariant, $R$ is $G_\mu$-invariant and so is $\Ro$. This proves $(i)$. To prove $(ii)$, we first observe that if $\bm{\alpha}_q\in J^{-1}(\bm{0})$ and $\pmb{v}_q\in J_L^{-1}(\bm{\mu})$ then 
\[
 \langle\alpha^a_{q},v_{aq}\rangle =  \langle\mathcal{T}^a(\bm{\alpha}_q),\tilde\tau^a(\pmb{v}_q)\rangle,
\]
where 
\[
\mathcal{T}^a\colon J^{-1}(\pmb{0}) \to T^*(Q/G), \quad  \tilde\tau^a \colon J_L^{-1}(\mu)  \to  T(Q/G),
\]
are the $a$-th components of the maps $\mathcal{T}$ and $\tilde\tau$, respectively. This is a direct consequence of the definitions of $\mathcal{T}$ and $\tilde\tau$. It follows that, if $\pmb{v}_q$ and $\pmb{w}_q$ are elements of $T^1_kQ$, then 
\[
\langle\mathcal{S}^a\circ (F L)^a(\pmb{v}),w_{aq}\rangle = \langle\mathcal{T}^a\circ \mathcal{S}\circ (\F L)(\pmb{v}),\tilde\tau^a(\pmb{w})\rangle ,
\]
where $\mathcal{S}^a\colon T^*Q\to T^*Q$ is the map $\mathcal{S}^a(\beta_{q})=\beta_q-\Ac_{\mu_a}$.

If $\tilde\tau(\pmb{v})$ and $\tilde\tau(\pmb{w})$ are arbitrary elements in $\pi^*\left(T^1_k(Q/G)\right)$, then using the fact that $\tilde\tau$ is linear on the fibers:
\begin{align*}
\langle(\F \Ro)^a(\tilde\tau(\pmb{v})),\tilde\tau^a(\pmb{w})\rangle&=\left.\frac{d}{ds}\right|_{s=0}\Ro(\tilde\tau^1(\pmb{v}),\dots,\tilde\tau^a(\pmb{v}+s\pmb{w}),\dots,\tilde\tau^k(\pmb{v}))\\ 
&=\left(\left.\frac{d}{ds}\right|_{s=0}L(v_{1q},\dots,v_{aq}+sw_{aq},\dots,v_{kq})\right)-\Ac_{\mu_a}(w_{aq})\\
&=\langle(F L)^a(\pmb{v})),w_{aq}\rangle-\Ac_{\mu_a}(w_{aq})
=\langle\mathcal{S}^a\circ (F L)^a(\pmb{v}),w_{aq}\rangle \\
&=\langle\mathcal{T}^a\circ \mathcal{S}\circ (\F L)(\pmb{v}),\tilde\tau^a(\pmb{w})\rangle=\langle\tau^a\circ \F L\circ(\tilde\tau)^{-1}(\tilde\tau(\pmb{v})),\tilde\tau^a(\pmb{w})\rangle,
\end{align*}
with $\tau^a\colon J^{-1}(\mu)\to T^*(Q/G)$ the $a$-th component of $\tau$. Hence $\F\Ro=\tau\circ \F L\circ(\tilde\tau)^{-1}$ as desired. The proof of $(iii)$ is similar.
\end{proof}

We have now the following diagram (see~Diagram~\eqref{dia:red})
\begin{equation}\label{dia:red2}
\begin{tikzcd}[column sep=large]
T_k^1Q \arrow[r,"\F L"]\arrow[d,"{\rm Poly-Red}"']  & (T_k^1)^*Q\arrow[d,"{\rm Poly-Red}"]\\
p_\mu^*\left(T^1_k(Q/G)\right)\arrow[r,"\F \Ro_\mu"]& p_\mu^*\left((T^1_k)^*(Q/G)\right)
\end{tikzcd}
\end{equation}
where both $\F L$ and $\F \Ro_\mu$ are polysymplectomophisms (we are using again Lemma~\ref{lem:lema1}). The arrows ``Poly-Red'' above account for polysymplectic reduction followed by an identification for each of the reduced spaces $J^{-1}(\mu)/G_\mu$ and $J_L^{-1}(\mu)/G_\mu$. 

We see in Diagram~\eqref{dia:red2} that the function $\Ro_\mu$ plays the role of a reduced Lagrangian in the space $p_\mu^*\left(T^1_k(Q/G)\right)$; we will give precise meaning to this analogy soon. So far, we have proved the following result:

\begin{theorem}\label{thm:reducedform} Let $L\colon T^1_k Q\to\R$ be a regular, $G$-invariant and $G$-regular Lagrangian and consider the polysymplectic manifold $(T^1_k Q,\omega^a_{Q,L})$. Let $\mu\in\lagdk$ be a regular value of the momentum map and fix a principal connection $\Ac$ on $\pi\colon Q\to Q/G$. 

Then the reduced polysymplectic space can be identified with 
 \[
p_\mu^*\left(T^1_k(Q/G)\right). 
 \]
The reduced polysymplectic forms are given, for each $a=1,\dots,k$, by
\begin{equation}\label{eq:reducedform}
\overline{\omega}^a_\mu= (\F \Ro_\mu)^*\left(({\rm pr}^a_2)^*\omega_{Q/G}-({\rm pr}_1)^*\Bc_{\mu_a}\right). 
\end{equation}
\end{theorem}

\subsection{The reduced Lagrangian field theory}

The goal now is to relate solutions of the original Lagrangian field theory with solutions of a reduced Lagrangian field theory. 

\begin{lemma}\label{lem:energy} The energy $E_L\colon T^1_kQ\to\R$ is $G_\mu$-invariant. Its reduction is the function $E_{\Ro_\mu}$ defined as follows:
\[
E_{\Ro_\mu}(\pmb{v})\equiv \langle \F \Ro_\mu(\pmb{v}),\pmb{v} \rangle -\Ro_\mu (\pmb{v}),\qquad  \pmb{v}\in p_\mu^*\left(T^1_k(Q/G)\right).
\]
\end{lemma}

\begin{proof} The fact that $E_L\colon T^1_kQ\to\R$ is invariant follows directly from the equivariance of the Legendre transformation $\F L$. For the second part, note that $E_L$ might be as written in the form 
 \[
 E_L(\pmb{w})=\langle \F L (\pmb{w}),\pmb{w}\rangle -L(\pmb{w})=\langle \F R (\pmb{w}),\pmb{w}\rangle -R(\pmb{w}),\qquad  \pmb{w}\in T^1_kQ.
 \]
Therefore its pullback to $\pi^*(T^1_k(Q/G))$ is
\[
\big((\tilde\tau^{-1})^*E_L\big)(\pmb{v})=\langle \F \Ro (\pmb{v}),\pmb{v}\rangle -\Ro(\pmb{v}),\qquad  \pmb{v}\in \pi^*(T^1_k(Q/G)).
\]
From here the claim follows directly.
\end{proof}

Theorem~\ref{thm:reducedform}, together with Lemma~\ref{lem:energy}, shows that the Routhian $\Ro_\mu$ plays the role of the Lagrangian function in the reduced Lagrangian field theory: it encodes both the polysymplectic form and the energy function. It should be noted, however, that the reduced polysymplectic form~\eqref{eq:reducedform} has an additional term, so the reduced equations of motion are not the usual EL equations. This additional term does not depend on the Routhian, but on the chosen principal connection  $\Ac$.

We can finally state the main reduction theorem:
\begin{theorem}\label{thm:red1} Let $L\colon T^1_kQ\to\R$ be an hyperregular, $G$-invariant and $G$-regular Lagrangian, and fix a principal connection $\Ac$ on $\pi\colon Q\to Q/G$. Then, if $\pmb{\Gamma}$ is a $G_\mu$-invariant solution of the $k$-symplectic Euler-Lagrange equations
\[
 \sum_a \Gamma_a\lrcorner \omega_{Q,L}^a=dE_L, 
\]
which is tangent to $J_L^{-1}(\mu)$, the reduced $k$-vector field $\overline{\pmb{\Gamma}}_\mu$ on $p_\mu^*(T^1_k(Q/G))$ satisfies the following $k$-symplectic system:
\[
\sum_a \left(\overline{\Gamma}_\mu\right)_a\lrcorner \overline{\omega}_\mu^a=dE_{\Ro_\mu},  
\] 
where $\overline{\omega}_\mu^a$ are given by~\eqref{eq:reducedform}. 
\end{theorem}

\begin{proof} This follows from Theorem~\ref{thm:polyred} taking into account the observations above, Theorem \ref{thm:reducedform} and Lemma \ref{lem:energy}. 
\end{proof}

 We recall that $\pmb{\Gamma}_\mu$ denotes the restriction of $\pmb{\Gamma}$ to $J_L^{-1}(\mu)$. With regards to reconstruction, our situation is a particular case of Theorem~\ref{thm:reconstruction}:

\begin{theorem}\label{thm:reconstruction2} Under the same conditions of Theorem~\ref{thm:red1}, let $\overline{\phi}_\mu\colon \R^k\to p_\mu^*(T^1_k(Q/G))$ be an integral section of $\overline{\Gamma}_\mu$ such that the connection $\mathcal{H}(\pmb{\Gamma}_\mu,\overline{\phi}_\mu)$ is flat. Then there exists an integral section $\phi_\mu\colon \R^k\to J_L^{-1}(\mu)$ of $\pmb{\Gamma}_\mu$ with $\pi_\mu\circ \phi_\mu=\overline{\phi}_\mu$.
\end{theorem}

\section{Examples} \label{ex2sec}
 
\subsection{Navier's equations} \label{exNE}

We come back to Navier's equations from Section~\ref{subsec:someexamples}. To ease the notation we will write $x\equiv t^1$ and $y\equiv t^2$ for the parameters of the field theory.

The Lagrangian~\eqref{eq:NavierLagrangian} has $q^1$ and $q^2$ as cyclic variables.
First, we will consider only   the translations in  one of the variables, say $q^1$, as the symmetry group.
 Then $\xi_Q=\partial/ \partial {q^1}$ and the momentum map $J_L\colon T^1_2\R\to \R^2$ is
\[
J_L(q^i,v^i_a)=\left(\fpd{L}{v^1_1},\fpd{L}{v^1_2}\right)= \left((\lambda+2\nu)v^1_1+(\lambda+\nu)v^2_2,\nu v^1_2\right). 
\]

If we fix a momentum $\mu=(\mu_1,\mu_2)$ and impose that each of the components of the momentum map is preserved we get the following relations:
\begin{equation}\label{eq:naviermomentum}
v^1_1=\frac{\mu_1-(\lambda+\nu)v^2_2}{\lambda+2\nu},\qquad v^1_2=-\frac{\mu_2}{\nu}. 
\end{equation}
Apart from the regularity conditions $\nu\neq 0$ and $2\lambda+3\nu\neq 0$, one needs to make the assumption that $\lambda+2\nu\neq0$ to ensure $G$-regularity. If we take the standard connection of the bundle $(x,y)\mapsto x$ we find  for the Routhian  $R\colon T^1_2\R^2\to \R$:
\[
R(q^i,v^i_a)= L-\mu_1 v^1_1-\mu_2 v^1_2.
\]
The restriction of $R$ to $J_L^{-1}(\mu)$ is:
\[
R(q^1,q^2,v^2_1,v^2_2)= \frac{\nu}{2(\lambda+2\nu)}\left[ (\lambda+2\nu)(v^2_1)^2+(2\lambda+3\nu)(v^2_2)^2\right]+ C_1v^2_2+C_2,
\]
where we have used~\eqref{eq:naviermomentum} to replace $v^1_1$ and $v^1_2$. Here $C_1$ and $C_2$ are constants depending on the values of $\mu_1$ and $\mu_2$. Since $G$ is Abelian, $G_\mu=G$ and the Routhian $\Ro_\mu\colon T^1_2\R\to\R$ has the same expression as $R$ (note, however, that it is defined on a different space).

A solution $\psi\colon \R^2\to \R$ of the Euler-Lagrange equations for $\Ro_\mu$, which can also be obtained in the polysymplectic framework since $\Ro_\mu$ is regular, satisfies the Laplace equation:
\[
(\lambda+2\nu)\psi_{xx}+(2\lambda+3\nu)\psi_{yy}=0. 
\]
Consider, for simplicity, the solution $\psi(x,y)=xy$. A solution $\phi=(\varphi,\psi)\colon \R^2\to \R^2$ of the original Euler-Lagrange equations for $L$ can be obtained integrating the momentum constraints~\eqref{eq:naviermomentum}, namely 
\[
\varphi_x=\frac{\mu_1-(\lambda+\nu)x}{\lambda+2\nu},\qquad \varphi_y=-\frac{\mu_2}{\nu}.
\]
This gives
\[
\phi=\left(C+\frac{\mu_1x-(\lambda+\nu)x^2}{2(\lambda+2\nu)}-\frac{\mu_2}{\nu}y,xy\right), 
\]
for some constant $C$. One can readily check that $\phi$ solves~\eqref{eq:Navier} (with $\phi^1=\varphi$ and $\phi^2=\psi$). The same procedure applies to any solution $\psi$ of $\Ro_\mu$ as long as there exists a solution $\phi$ of the original Lagrangian projecting onto $\psi$. On the other hand, for the solution 
\[
\psi(x,y)=\cos\left(x\right) \cosh\left(\sqrt{\frac{\lambda+2\nu}{2\lambda+3\nu}}\;y\right)
\]
the system~\eqref{eq:naviermomentum} is inconsistent, since $\varphi_{yx}$ vanishes while $\varphi_{xy}$ does not. Therefore $\psi$ is not the projection of an invariant solution $\phi$ of $L$ tangent to the level set of $\mu$. We remind the reader that solutions $\psi$ of $\Ro_\mu$ that lift to a solution of $L$ have been characterized in Theorem~\ref{thm:reconstruction2}. 

Finally, we remark that there also exists solutions $\phi=(\varphi,\psi=xy)$ which can not be retrieved through the reconstruction procedure. In other words, even when solutions $\phi$ of $L$ exist which project onto $\psi$, it is not possible to reconstruct all of them through the procedure above. Consider for instance the solution
\[
\phi=\left(y^2-\frac{(\lambda+3\nu) x^2}{2(\lambda+2\nu)},xy\right) 
\]
of~\eqref{eq:Navier}. This can not be obtained integrating from~\eqref{eq:naviermomentum} for any choice of $\mu=(\mu_1,\mu_2)$ . This is to be expected because, for this solution, Noether's Theorem reads
\[
\fpd{}{x}\left(\fpd{L}{v^1_1}\circ \phi^{(1)}\right) + \fpd{}{y}\left(\fpd{L}{v^1_2}\circ \phi^{(1)}\right)=(\lambda+2\nu)\varphi_{xx}+(\lambda+\nu)\psi_{xy}+\nu\varphi_{yy}=0,
\]
and none of the components of $J_L$ is constant. For example, the first component of $J_L$ is not preserved along $\phi$:
\[
\fpd{L}{v^1_1}\circ \phi^{(1)}= (\lambda+2\nu)\phi_x+(\lambda+\nu)\psi_y=-2\nu x.
\]
We refer the reader to the discussion  after Remark~\ref{remark1}.

The Lagrangian of a complex scalar field~\eqref{eq:complexsf} is invariant under rotations in the $(\phi_1,\phi_2)$ plane.
 When expressed in polar coordinates, it is another example of a Lagrangian with a cyclic coordinate and one may proceed in a similar way as  above, for Navier's equations.

\begin{remark} Recall that the Lagrangian (\ref{eq:NavierLagrangian}) had two cyclic variables, and we could therefore have  reduced the system by both variables. However, it is often  not desirable to exhaust  the whole symmetry group, because it may limit the amount of suitable reduced solutions. In the particular case of Navier's equations,   one finds only  constant solutions when one uses the full symmetry group.  The same observation applies to the Laplace equation (\ref{Lapeq}).\end{remark}

\subsection{The case of a Lie group}\label{sec:exLie}

Given a Lie group $G$, we consider a (regular and $G$-regular) left invariant  Lagrangian $L\colon T^1_kG\to \R$. This general situation appears, for example, when one studies Lagrangians given by invariant metrics on $G$. In the mechanical case ($k$=1), this example is discussed in detail in \cite{LGC_class}. 

We first recall some well-known facts. When $G$ acts on itself by left translation, the principal bundle of interest is simply the projection of $G$ onto the neutral element $\{e\}$ (or any other point). The infinitesimal generators are the right-invariant vector fields, and every tangent vector $v_g\in G$ is vertical. This bundle admits a canonical flat principal connection given by $\Ac(v_g)=v_g\cdot g^{-1}$. If $\nu\in\lag^*$, the contraction $\Ac_\nu=\langle\nu,\Ac\rangle$ satisfies 
\begin{equation*}
d\Ac_\nu (g\cdot\xi,g\cdot\eta)=\langle\nu,[{\rm Ad}_g\xi,{\rm Ad}_g\eta]\rangle=\langle {\rm Ad}_g^*\nu,[\xi,\eta]\rangle, 
\end{equation*}
where we have used Cartan's structure equations.

We identify on the left $TG\to G\times\lag$, $v_g\mapsto g^{-1}\cdot v_g$, so that $L$ may be thought of as a function of $G\times \lagk$. Invariance then implies that $L$ is of the form
\[
L(g,g\cdot\xi_1,\dots,g\cdot\xi_k)=\ell(\xi_1,\dots,\xi_k)
\]
where $\ell\colon\lagk\to \R$. If we denote by $\pmb{F} \ell\colon\lagk\to\lagdk$ the fiber derivative of $\ell$, with components $(F\ell)^a\colon \lagk\to \lag^*$ given by
\[
\langle (F\ell)^a(\xi_1,\dots,\xi_k),\eta\rangle=\left.\frac{d}{ds}\right|_{s=0}\ell(\xi_1,\dots,\xi_a+s\eta,\dots,\xi_k), 
\]
then one finds
\begin{equation}\label{eq:ex1}
J_L(g,g\cdot\xi_1,\dots,g\cdot\xi_k)=\big({\rm Ad}^*_{g^{-1}}\big)^k\pmb{F} \ell(\xi_1,\dots,\xi_k).  
\end{equation}
One can prove~\eqref{eq:ex1} directly applying the reasoning in the mechanical case to the components $J_L^a$ of the momentum map $J_L$, see \cite{LGC_class}. In view of \eqref{eq:ex1}, fixing a momentum $\mu=(\mu_1,\dots,\mu_k)$ we find:
\begin{equation}\label{eq:ex3}
J_L(g,g\cdot\xi_1,\dots,g\cdot\xi_k)=\mu\iff  \pmb{F} \ell(\xi_1,\dots,\xi_k)=\big({\rm Ad}^*_{g}\big)^k\mu=({\rm Ad}^*_{g}\mu_1,\dots,{\rm Ad}^*_{g}\mu_k). 
\end{equation}

We will use the following short notation. We let $\nu\in\lagdk$ be
\[
\nu\equiv (\nu_1,\dots,\nu_k)= ({\rm Ad}^*_{g}\mu_1,\dots,{\rm Ad}^*_{g}\mu_k)\equiv \big({\rm Ad}^*_{g}\big)^k\mu.
\]
Since every element on $TG$ is vertical w.r.t. the projection $G\to\{e\}$, from the last relation one checks easily that $G$-regularity implies that $\pmb{F} \ell$ is a diffeomorphism. We shall write
\[
\tau\equiv (\pmb{F} \ell)^{-1}\colon \lagdk\to\lagk
\]
to denote its inverse. The relation \eqref{eq:ex3} of the momentum constraint is then 
\[
\pmb{F} \ell(\xi_1,\dots,\xi_k)=\big({\rm Ad}^*_{g}\big)^k\mu=\nu\iff (\xi_1,\dots,\xi_k)=\tau(\nu).
\]

We need to recall the definition of the $k$-coadjoint orbit through $\mu\in\lagdk$, which is a polysymplectic manifold. We will denote it by $\mathcal{O}_\mu =\mathcal{O}_{(\mu_1,\dots,\mu_k)}$ and it is defined as the orbit of the ${\rm Coad}^k$-action, i.e.
\[
\mathcal{O}_\mu=\{g\cdot \mu\st g\in G\}\subseteq\lagdk. 
\]
It has been shown that $G/G_\mu\simeq \mathcal{O}_\mu$, see~\cite{Polyreduction}.

From the general considerations in this paper (in particular, Proposition~\ref{prop:identificationJL}), applied to the present case $Q=G$, we have $J_L^{-1}(\mu)/G_\mu\simeq G/G_\mu\simeq \mathcal{O}_\mu$. We will now describe how the Routhian $\Ro_\mu\colon \mathcal{O}_\mu\to \R$ is obtained. We will use the $\Ac$ introduced before. The Routhian $R\colon G\times \lagk\to \R$ is then:
\begin{align*}
R(g,g\cdot\xi_1,\dots,g\cdot\xi_k)&=L(g,g\cdot\xi_1,\dots,g\cdot\xi_k)- \langle \mu_1,\Ac(g\cdot\xi_1)\rangle - \dots - \langle \mu_k,\Ac(g\cdot\xi_k)\rangle\\
&=\ell(\xi_1,\dots,\xi_k)-\langle {\rm Ad}_g^*\mu_1,\xi_1\rangle - \dots - \langle {\rm Ad}_g^*\mu_k,\xi_k\rangle\\
&=\ell(\xi_1,\dots,\xi_k)-\langle \nu_1,\xi_1\rangle - \dots - \langle \nu_k,\xi_k\rangle.
\end{align*}
The restriction of $R$ to $J_L^{-1}(\mu)$, denoted $\Ro\colon G\to \R$ as in Section \ref{sec:Routh}, is then
\begin{equation}\label{eq:ex2}
\Ro(g)= \left.\left(\ell(\xi_1,\dots,\xi_k)-\langle \nu_1,\xi_1\rangle - \dots - \langle \nu_k,\xi_k\rangle\right)\right|_{(\xi_1,\dots,\xi_k)=\tau(\nu)}.
\end{equation}
We make an observation before moving on. The notation in \eqref{eq:ex2} is meant to imply that the occurrences of $\xi_1,\dots,\xi_k$ are replaced in terms of $\nu_1,\dots,\nu_k$ using the map $\tau$. Thus, the momentum map fixes the ``$\lagk$-component'' of $L$, and as a result $\Ro$ is a function of $G$ alone. In other words, the level set of $\mu$ is identified with $G$.

Proceeding with the reduction by $G_\mu$, we finally find that the (reduced Routhian) $\Ro_\mu\colon \mathcal{O}_\mu\to \R$ is
\[
\Ro_\mu(\nu_1,\dots,\nu_k)= \left.\left(\ell(\xi_1,\dots,\xi_k)-\langle \nu_1,\xi_1\rangle - \dots - \langle \nu_k,\xi_k\rangle\right)\right|_{(\xi_1,\dots,\xi_k)=\tau_\mu(\nu)},
\]
with $\tau_\mu$ the restriction of $\tau$ to $\mathcal{O}_\mu$. We have used that the diffeomorphism $G/G_\mu\to \mathcal{O}_\mu$ is given by $g\mapsto g\cdot \mu$.

\begin{remark} If one wants to compute the equations of motion one needs to take into account the force terms $\Bc_{\mu_a}$. Following \cite{LGC_class}, one can relate these terms to the the Kirillov–Kostant–Souriau symplectic forms on $\mathcal{O}_{\mu_a}\subseteq \lag^*$.  
\end{remark}

A particular case of the construction in the previous example is obtained when the Lagrangian $L$ is constructed from a left-invariant metric $\mathcal{G}$ on $G$. Let us write
\begin{equation}\label{eq:ex4}
L(g,g\cdot\xi_1,\dots,g\cdot\xi_k)=\frac{1}{2}\mathcal{G}(\xi_1,\xi_1)+\dots+\frac{1}{2}\mathcal{G}(\xi_k,\xi_k)=\ell(\xi_1,\dots,\xi_k). 
\end{equation}
The expression of $L$ makes it simple to compute the maps $\pmb{F}\ell$ and $\tau$. First, the fiber derivative $\pmb{F}\ell$ has components
\[
(F\ell)^1(\xi_1,\dots,\xi_k)=\flat(\xi_1),\quad (F\ell)^2(\xi_1,\dots,\xi_k)=\flat(\xi_2),\quad \dots,
\]
where $\flat\colon\lag\to\lag^*$ is the isomorphism induced by the metric $\langle \flat(\xi),\eta\rangle=\mathcal{G}(\xi,\eta)$. Therefore the inverse $\tau\colon \lagdk\to \lagk$ has components:
\[
(\tau)^1(\nu_1,\dots,\nu_k)=\flat^{-1}(\nu_1),\quad (\tau)^2(\nu_1,\dots,\nu_k)=\flat^{-1}(\nu_2),\quad \dots.
\]

Let us denote by $\widetilde{\mathcal{G}}=(\flat^{-1})^*(\mathcal{G})$ the metric induced on $\lag^*$. The Routhian is then
\begin{align*}
\Ro_\mu(\nu)&= \left.\left(\frac{1}{2}\mathcal{G}(\xi_1,\xi_1)+\dots+\frac{1}{2}\mathcal{G}(\xi_k,\xi_k)-\langle \nu_1,\xi_1\rangle - \dots - \langle \nu_k,\xi_k\rangle\right)\right|_{\xi_a=\flat^{-1}(\nu_a)}\\ 
&=\frac{1}{2}\widetilde{\mathcal{G}}(\nu_1,\nu_1)+\dots+\frac{1}{2}\widetilde{\mathcal{G}}(\nu_k,\nu_k)-\langle \nu_1,\flat^{-1}(\nu_1)\rangle - \dots - \langle \nu_k,\flat^{-1}(\nu_k)\rangle\\
&=\frac{1}{2}\widetilde{\mathcal{G}}(\nu_1,\nu_1)+\dots+\frac{1}{2}\widetilde{\mathcal{G}}(\nu_k,\nu_k)-\widetilde{\mathcal{G}}(\nu_1,\nu_1)-\dots-\widetilde{\mathcal{G}}(\nu_k,\nu_k)\\
&=-\left(\frac{1}{2}\widetilde{\mathcal{G}}(\nu_1,\nu_1)+\dots+\frac{1}{2}\widetilde{\mathcal{G}}(\nu_k,\nu_k)\right).
\end{align*}
This coincides, up to a sign, with the reduced Hamiltonian $H_\mu$ in the Example 4.3.1 of \cite{Polyreduction}. This is to be expected: since there are no fiber coordinates $\pmb{v}$ in $\Ro_\mu$, the usual ansatz to compute the Hamiltonian gives
\[
H_\mu(q,\pmb{\alpha})=\left(\langle\F \Ro_\mu(\pmb{v}),\pmb{v}\rangle-\Ro_\mu(q,\pmb{v})\right)_{\pmb{v}=(\F \Ro_\mu)^{-1}(\pmb{\alpha})}=-\Ro_\mu(q,\pmb{v}).
\]

The case of an invariant metric on a Lie group above, where the Lagrangian is given by \eqref{eq:ex4}, is of interest because it admits a solution $\pmb{\Gamma}$ which can be reduced to a solution of the Routhian $\Ro_\mu$ (in the sense of Theorem~\ref{thm:red1}). The explicit construction of such a solution $\pmb{\Gamma}$ can be done adapting the Hamiltonian version in Example 4.3.1 of \cite{Polyreduction} mentioned above. From a Lagrangian perspective, we will see in the next example that such a solution $\pmb{\Gamma}$ appears naturally.

\subsection{Invariant metrics and harmonic maps} \label{exHM}

Let us first recall some well-known facts.  Consider the natural Lagrangian $\tilde 
L\colon TQ\to \R$  that one may associate to a Riemannian metric $\mathcal{G}$ on $Q$:
\begin{equation}\label{eq:ex7}
\tilde {L}(q,v)=\frac{1}{2}\mathcal{G}(v,v)=\frac{1}{2}g_{ij}(q)v^iv^j. 
\end{equation}
This Lagrangian is regular, and its unique solution is the geodesic spray 
\[
\Gamma_{\tilde L}=v^j\fpd{}{q^j}-\Gamma^i_{jk}v^jv^k\fpd{}{v^i}, 
\]
where $\Gamma^i_{jk}$ are the Christoffel symbols of $\mathcal{G}$. If $\mathcal{G}$ is invariant under a $G$-action, then the Lagrangian is also $G$-invariant and $\Gamma_{\tilde L}$ is $G$-invariant. The second-order vector field $\Gamma_{\tilde L}$ is completely determined from the relation $i_{\Gamma_{\tilde{L}}}\omega_{Q,\tilde{L}}=dE_{\tilde L}$.

We now consider a Lagrangian $L\colon T^1_kQ\to\R$ of the form
\begin{equation}\label{eq:ex5}
L(v_1,\dots,v_k)=\frac{1}{2}\mathcal{G}(v_1,v_1)+\dots+\frac{1}{2}\mathcal{G}(v_k,v_k),  
\end{equation}
where $\mathcal{G}$ is a $G$-invariant metric on $Q$ (the previous example falls in this category when $Q=G$). The form of the Lagrangian~\eqref{eq:ex5} is motivated by the following well-known observation that we outlined in Section~\ref{subsec:someexamples}: the defining relation for a harmonic map between an Euclidean space and a Riemannian manifold $(Q,{\mathcal G})$ can be thought of as a Lagrangian field equation, when the Lagrangian of the field theory is exactly of the form~\eqref{eq:ex5}. We will now show that a Lagrangian of that form admits an invariant harmonic as solution.

Let us write, for simplicity, 
\[
\xi_Q=\Lambda^j\fpd{}{q^j} 
\]
for an infinitesimal generator of the action, where $\Lambda^j$ are functions on $Q$ determined from the map $\xi\mapsto \xi_Q$. Since $G$ acts on the left by isometries, each of these generators is a Killing vector field of the metric $\mathcal{G}$ and satisfies the Killing equation ${\mathcal L}_{\xi_Q} \mathcal{G}=0$, or
\begin{equation}\label{eq:killing}
\Lambda^j \fpd{g_{ki}}{q^j} + g_{ji}\fpd{\Lambda^j}{q^k} + g_{kj}\fpd{\Lambda^j}{q^i} = 0.
\end{equation}
The infinitesimal generators of the lifted action to $T^1_kQ$ are the complete lifts of $\xi_Q$:
\[
\clift{\xi_Q}= \Lambda^j\fpd{}{q^j}+\sum_b v^j_b \fpd{\Lambda^i}{q^j}\fpd{}{v^i_b}. 
\]
We also note that for the Lagrangian~\eqref{eq:ex5} the components of the momentum map are
\begin{equation}\label{eq:momentumlocked}
(J_L)_b^\xi=\mathcal{G}(v_b,\xi_Q)=g_{ij}v^i_b \Lambda^j. 
\end{equation}

Consider now the $k$-vector field $\pmb{\Gamma}$ on $T^1_kQ$ with components 
\begin{equation*}
\Gamma_a=v^i_a\fpd{}{q^i}- \Gamma^i_{jk}v^j_a \sum_b v^k_b\fpd{}{v^i_b}. 
\end{equation*}
One may readily verify that $\pmb{\Gamma}$ is a solution of \eqref{eq:k-EL} for the Lagrangian \eqref{eq:ex5}. Moreover, it is is $G$-invariant (this can be done, for instance, by verifying that the brackets $[\Gamma_a,\clift{\xi_Q}]$ vanish). Further, by making use  of the defining relation of the Christoffel symbols,
\[
\Gamma^p_{ki}  g_{p j} = \frac12 \fpd{g_{jk}}{q^i}+ \frac12\fpd{g_{ji}}{q^k} - \frac12 \fpd{g_{ik}}{q^j},
\]
and the Killing equation~\eqref{eq:killing}, we obtain the following expression for ($\Gamma_a\left((J_L)_b^\xi\right)$):
\[
\Gamma_a\left((J_L)_b^\xi\right) =v^k_bv_a^i\left(\frac 12 \fpd{g_{ij}}{q^k}\Lambda^j+ \frac12 g_{ij}\fpd{\Lambda^j}{q^k}-\frac12 \fpd{g_{jk}}{q^i} \Lambda^j - \frac12  g_{kj}\fpd{\Lambda^j}{q^i} \right).
\]
The expression above is of the type $v^i_a v^k_b S^j_{ik}$, where $S^j_{ik}$ is skew-symmetric in $i$ and $k$. An obvious case where this expression vanishes is that when  $v^i_a=v^i_b\equiv v^i$ for each $i$. This case occurs, e.g., when each of the components of the momentum map are equal, i.e.\ when $\mu=(\nu,\dots,\nu)$ for some $\nu\in\lag^*$ and when expression \eqref{eq:momentumlocked} has a unique solution $v^i$.  After briefly describing the effective computation of the Routhian for the Lagrangian~\eqref{eq:ex5} below, we give a concrete example where  we can apply the Routhian technique.

Let us choose the mechanical connection $\Ac$ on $Q\to Q/G$. Recall that vector $v_q\in TQ$ is horizontal for $\Ac$ if the following condition holds:
\[
v_q\; \text{is horizontal} \iff  \mathcal{G}(v_q,\xi_Q(q))=0,\; \text{for all}\; \xi\in\lag.
\]
In other words, the horizontal subspace of $\Ac$ is the orthogonal (w.r.t. the metric $\mathcal{G}$) of the vertical subbundle $V\pi$. The use of this principal connection greatly simplifies the computation of the Routhian. Our approach here follows \cite{RouthMarsden}.

It is customary to call the map $I_q\colon \lag\to\lag^*$ defined by
\[
\langle I_q(\xi),\eta\rangle=\mathcal{G}(\xi_Q(q),\eta_Q(q)) 
\]
the locked inertia tensor; we remark that, in general, it depends on the chosen point $q\in Q$. Let us denote $\zeta_a=\Ac(v_a)\in\lag$. When $\mu=J_{L}(\pmb{v})$, from \eqref{eq:momentumlocked} we have
\[
\mathcal{G}({\rm Ver}(v_a),{\rm Ver}(v_a))= \mathcal{G}\big((\zeta_a)_Q(q),(\zeta_a)_Q(q)\big)=\langle I_q(\zeta_a),\zeta_a\rangle=\langle \mu_a,\zeta_a\rangle,
\]
and we can write $\mathcal{G}({\rm Ver}(v_a),{\rm Ver}(v_a))=\langle\mu_a,I_q^{-1}(\mu_a)\rangle$. Then, if for each $v_a$ we decompose it in its horizontal and vertical parts, we have for the Routhian $\Ro\colon J_L^{-1}(\mu)\to\R$:
\begin{align}\label{eq:ex6}
\Ro(q,v_1,\dots,v_k)&=\sum_a \left[\frac{1}{2}\mathcal{G}({\rm Hor}(v_a),{\rm Hor}(v_a)+ \frac{1}{2}\mathcal{G}({\rm Ver}(v_a),{\rm Ver}(v_a)-\langle\mu_a,\zeta_a\rangle\right]\nonumber \\ 
&=\sum_a\left[\frac{1}{2}\mathcal{G}({\rm Hor}(v_a),{\rm Hor}(v_a)- \frac{1}{2}\langle\mu_a,I_q^{-1}(\mu_a)\rangle\right].
\end{align}
We remark that the terms $\langle\mu_a,I_q^{-1}(\mu_a)\rangle$ are $G_\mu$-invariant, so $\Ro$ can be easily reduced to define the function $\Ro_\mu$ on the quotient. Note that \eqref{eq:ex6} is the field-theoretic analogue of the expression in Proposition 3.5 in \cite{RouthMarsden}. Of course, such a simple expression can only be obtained because the Lagrangian \eqref{eq:ex5} is a sum of Lagrangians on $TQ$, pulled-back to $T^1_kQ$.

{\bf A concrete example.}  The following example also appeared in the context of Lagrange-Poincar\'e reduction in the $k$-symplectic formalism \cite{LTM_LP}. Consider the Lie group $G$ of matrices of the form:
\[
\begin{pmatrix}
1 & y\cos\theta+x\sin\theta & -y\sin\theta+x\cos\theta & z\\
0 & \cos\theta & -\sin\theta & x\\
0 & \sin\theta & \cos\theta & -y\\
0 & 0 &0 & 1
\end{pmatrix}. 
\]
For details on $G$, we refer the reader to \cite{Thompson} (case ``$A_{4,10}$'' on page 423) or \cite{LTM_LP}. Take $Q=\R\times G$ and consider the Lagrangian $L\colon T^1_kQ\to\R$ given by \eqref{eq:ex7} where the metric in $Q$ is:
\[
\mathcal{G}=dq^2+\gamma\,dq\,d\theta+dx^2+dy^2-y\,dx\,d\theta+x\,dy\,d\theta+dz\,d\theta, 
\]
with $dx\,d\theta=(dx\otimes d\theta+d\theta\otimes dx)/2$ the usual symmetric product notation (and similarly for the other terms) . The natural left action of $G$ on $Q$ gives isometries.

We take the following square matrices in dimension 4 as the basis of the Lie algebra $\lag$ of $G$:
\[
e_x=e_{13}-e_{24},\quad  e_y=e_{12}-e_{34},\quad e_z=e_{14},\quad e_{\theta}=-e_{23}+e_{32}.
\]
The notation $e_{ij}$ denotes the matrix with one entry equal to one at $(i,j)$ (and zero otherwise). The Lie algebra structure of $\lag$ is given by the following nonvanishing brackets: 
\begin{equation}\label{eq:brackets}
[e_x,e_y]=-2e_z,\qquad [e_x,e_\theta]=e_y,\qquad [e_y,e_\theta]=-e_x. 
\end{equation}
With this choice of basis, the infinitesimal generators are:
\[
E_x= \fpd{}{x}-y\fpd{}{z},\quad E_y= \fpd{}{y}+x\fpd{}{z},\quad E_z=\fpd{}{z},\quad E_\theta= \fpd{}{\theta}-x\fpd{}{y}+y\fpd{}{x}.
\]
These can be easily obtained since the infinitesimal generators for the left translation on a Lie group are precisely the right-invariant vector fields. Their brackets satisfy their lowercase analogs except for a sign; for instance, we have $[E_x,E_y]=2E_z$. 

A set of left-invariant vector fields has been obtained in~\cite{LTM_LP}:
\begin{align*}
F_x&= \cos\theta\left(\fpd{}{x}+y\fpd{}{z}\right)-\sin\theta\left(\fpd{}{y}-x\fpd{}{z}\right),\quad F_z=\fpd{}{z},\\
F_y&= \sin\theta\left(\fpd{}{x}+y\fpd{}{z}\right)+\cos\theta\left(\fpd{}{y}-x\fpd{}{z}\right),\quad F_\theta=\fpd{}{\theta}. 
\end{align*}

We consider the following vector field on $Q$:
\[
X=\fpd{}{x}-\gamma\fpd{}{q}. 
\]
The definition of $X$ is such that $\mathcal{G}(X,E_a)=0$ for all $a=x,y,z,\theta$, so that $X$ represents a horizontal vector field projecting onto $\partial/\partial q$ orthogonal to the vertical subspace, or in other words $X$ encodes the mechanical connection for the metric $\mathcal{G}$. This is a flat connection.

It is clear that $\mathcal{G}(X,X)=1$. With regards to the vertical part of the metric, if we compute it in the basis of infinitesimal generators we find:
\[
\mathcal{F}_{ab}\equiv\mathcal{G}(E_a,E_b)= 
\begin{pmatrix}
1 & 0 & 0 & 0\\
0 & 1 & 0 & 0\\
0 & 0 & 0 & 1/2\\
0 & 0 & 1/2 & 0
\end{pmatrix}.
\]
One shows in a similar way that $\mathcal{G}(F_a,F_b)=\mathcal{F}_{ab}$. In other words, the vertical part of the metric is a bi-invariant metric on $G$ (one can show as well that $\mathcal{G}$ is a bi-invariant metric on $Q$ when the natural action of $G$ on the right is considered). But then the locked inertia tensor $I_q$ does not depend on $q$; the matrix $\mathcal{F}$, which gives the metric at the Lie algebra level, is such that
\[
\langle I_q(\xi),\eta\rangle= \sum_{a,b}\xi^a\mathcal{F}_{ab}\eta^b
\]
when we use the basis $\{e_x,e_y,e_z,e_\theta\}$ of $\lag$. Looking at \eqref{eq:ex7} we see that all of the terms 
\[
-\frac{1}{2}\langle\mu_a,I_q^{-1}(\mu_a) \rangle
\]
are constant and do not contribute to the Euler-Lagrange equations of $\Ro$, so we may ignore them (but keep the name for $\Ro$, for simplicity). One concludes that the reduced Routhian $\Ro_\mu\colon Q/G_\mu\times T^1_k\R\to\R$ is given by the following expression:
\[
\Ro_\mu([q,g],v_1,\dots,v_k)= \left[\frac{1}{2}\mathcal{G}({\rm Hor}(v_1),{\rm Hor}(v_1)\right]+\dots+\left[\frac{1}{2}\mathcal{G}({\rm Hor}(v_k),{\rm Hor}(v_k)\right].
\]
 The coordinate formula for $\Ro_\mu$ is then
\[
\Ro_\mu(v_a)= \frac{1}{2}\Big((v^1_1)^2+\dots+(v^1_k)^2\Big).
\]
It is independent of the choice of $\mu\in\lagdk$ or the base point $(q,[g])\in Q/G_\mu$. As a matter of fact, one could still use the residual symmetry in $q$  to further reduce: the reason, as mentioned above, is that the metric $\mathcal{G}$ is bi-invariant for the Lie group $Q=\R\times G$ with product $(q_1,g_2)\cdot (q_2,g_2)=(q_1+q_2,g_1g_2)$. Reducing by the whole of $Q$ directly would take us to the results in the previous example.

Let us choose $\nu=(0,0,1,0)=e^z$, where $e^z$ is the dual of $e_z$). A computations shows that $G_\nu$ is the following Abelian subgroup of $G$:
\[
G_\nu=\left\{
\begin{pmatrix}
1 & 0 & 0 & z\\
0 & \cos\theta & -\sin\theta & 0\\
0 & \sin\theta & \phantom{-}\cos\theta & 0\\
0 & 0 &0 & 1
\end{pmatrix},\; z\in\R,\theta\in [0,2\pi)
\right\} \simeq \R\times S^1.
\]
One can check this directly using the infinitesimal condition for $\xi\in G_\nu$: an element $\xi\in\lag_\nu$ satisfies ${\rm ad}_\xi^*\nu=0$ or $\langle \nu,[\xi,\eta]\rangle=0$ for $\eta\in\lag$ arbitrary. From the adjoint relations in~\eqref{eq:brackets} we have 
\[
2(\xi^y\eta^x-\xi^x\eta^y)=0 
\]
where we have written $\xi=\xi^x e_x+\xi^y e_y+\xi^z e_z+\xi^\theta e_\theta\in\lag$ (and similar for $\eta$). This means that $\xi^x=\xi^y=0$, and therefore $\lag_\nu={\rm span}\{e_z,e_\theta\}$, which corresponds to the subgroup $G_\nu$ above.

We carry out the reduction at $\mu= (\nu,\dots,\nu)\in\lagdk$. Then $G_\mu=G_\nu=G_\nu=\R\times S^1$ and $\Ro_\mu$ is defined on $ Q/(\R\times S^1)\times T^1_k\R$. The reduced $k$-vector field 
\[
\overline{\Gamma}_a=v^1_a\fpd{}{q}
\]
is a solution of~\eqref{eq:k-EL} for $\Ro_\mu$. There is no force term because the base of $Q\to Q/G=\R$ is one-dimensional, and this implies that all of the $B_\mu$ (which are 2-forms on the base) must vanish. It is one of the benefits of our approach that finding solutions of the reduced Routhian equations is very simple in this case. The reconstruction of the integral sections of the original problem, however, involves solving a PDE problem, which is not a trivial matter. 


\paragraph{Acknowledgments.} We would like to thank J.C.\ Marrero for some clarifications concerning reference~\cite{Polyreduction}. S.\ Capriotti, V. D\'iaz and E.\ Garc\'{\i}a-Tora\~{n}o Andr\'{e}s are thankful to FONCYT for funding through project PICT 2019-00196. T.\ Mestdag thanks the Research Foundation -- Flanders (FWO) for its support through Research Grant 1510818N.  We would like to thank the referees for their valuable comments and recommendations to improve our paper.

\bibliographystyle{plain}

\end{document}